%% file: s-verifywithrl.tex
\documentclass[conference]{IEEEtran}
\IEEEoverridecommandlockouts

\usepackage{cite}
\usepackage{ dsfont }
\usepackage{amsmath}
\usepackage{cleveref}
\usepackage{graphicx}
\usepackage{enumerate}
\usepackage{xspace}
\usepackage{amsthm}
\usepackage{anyfontsize}
\usepackage{amsmath}
\usepackage{algorithm}
\usepackage[noend]{algpseudocode}

\usepackage{cases,color,array,url}
\usepackage{algorithmicx}
\usepackage{graphicx}
\usepackage{mdwmath}
\usepackage{mdwtab}
\usepackage{eqparbox}
\usepackage{amsopn}
\usepackage{multirow}
\usepackage{graphics}
\usepackage{url}
\usepackage{paralist}
\usepackage{xspace}
\usepackage{amsmath}
\usepackage{upgreek}

\usepackage{latexsym}
\usepackage{amsfonts}
\usepackage{amssymb}
\usepackage{comment}
\usepackage{mathrsfs}
\usepackage{cases}
\usepackage{listings}

\newcommand{\kwl}[1]{\mathbf{#1}}

\newcommand{\kwe}[1]{\mathsf{#1}}

\def \ourtool{SmartVerif\xspace}

\def\ie{\textit{i.e.}\xspace}
\def\etal{\textit{et al.}\xspace}
\def\etc{\textit{etc.}\xspace}

\def\eg{\textit{e.g.}\xspace}

\Crefformat{equation}{#2Eq.~(#1)#3}
\Crefmultiformat{equation}{Eqs.~(#2#1#3)}%
{ and~(#2#1#3)}{, (#2#1#3)}{ and~(#2#1#3)}

\newtheorem*{theorem}{Theorem}

\begin{document}
\title{Verifying Security Protocols using Dynamic Strategies}

\author{\IEEEauthorblockN{Yan Xiong}
\IEEEauthorblockA{\textit{School of Computer Science and Technology} \\
\textit{University of Science and Technology of China}\\
Hefei, China \\
yxiong@ustc.edu.cn}
\and
\IEEEauthorblockN{Cheng Su}
\IEEEauthorblockA{\textit{School of Computer Science and Technology} \\
\textit{University of Science and Technology of China}\\
Hefei, China \\
gotzeus@mail.ustc.edu.cn}
\and
\IEEEauthorblockN{Wenchao Huang}
\IEEEauthorblockA{\textit{School of Computer Science and Technology} \\
\textit{University of Science and Technology of China}\\
Hefei, China \\
huangwc@ustc.edu.cn}
\and
\IEEEauthorblockN{Fuyou Miao}
\IEEEauthorblockA{\textit{School of Computer Science and Technology} \\
\textit{University of Science and Technology of China}\\
Hefei, China \\
mfy@ustc.edu.cn}
\and
\IEEEauthorblockN{Wansen Wang}
\IEEEauthorblockA{\textit{School of Computer Science and Technology} \\
\textit{University of Science and Technology of China}\\
Hefei, China \\
wangws@mail.ustc.edu.cn}
\and
\IEEEauthorblockN{Hengyi Ouyang}
\IEEEauthorblockA{\textit{School of Computer Science and Technology} \\
\textit{University of Science and Technology of China}\\
Hefei, China \\
oyhy5995@mail.ustc.edu.cn}
}

\maketitle

\input{s-abstract}

\begin{IEEEkeywords}
Formal verification, Machine learning, Artificial intelligence
\end{IEEEkeywords}


\input{s-introduction}
\input{s-relatedwork}
\input{s-preliminaries}
\input{s-framework2-hwc}

\input{s-example}
\input{s-constructtree}
\input{s-dqn-sc}
\input{s-evaluation}
\input{s-futurework}

\input{s-conclusion}

\bibliographystyle{IEEEtran}
\bibliography{bibfile}
\newpage
\appendix
\input{s-a-proof}

\input{s-a-technique}

\input{s-a-casestudy}

\end{document}

%% file: s-abstract.tex

\begin{abstract}
Current formal approaches have been successfully used to find design flaws in many security protocols. 
However, it is still challenging to automatically analyze protocols due to their large or infinite state spaces.
In this paper, we propose \ourtool, a novel and general framework that pushes the limit of automation capability of state-of-the-art verification approaches.
The primary technical contribution is the \textit{dynamic} strategy inside \ourtool, which can be used to smartly search proof paths.
Different from the non-trivial and error-prone design of existing static strategies, the design of our dynamic strategy is simple and flexible: 
it can automatically optimize itself according to the security protocols without any human intervention.
With the optimized strategy, \ourtool localizes and proves \textit{supporting lemmata}, which leads to higher probability of success in verification.
The insight of designing the strategy is that the node representing a supporting lemma is on an incorrect proof path with lower probability, when a random strategy is given.
Hence, we implement the strategy around the insight by introducing a reinforcement learning algorithm.
We also propose several methods to deal with other technical problems in implementing \ourtool.
Experimental results show that \ourtool automatically verify all the studied security protocols which include protocols that cannot be automatically verified by existing approaches. 
The case study also validates the efficiency of our dynamic strategy.



\end{abstract}

%% file: s-introduction.tex
\section{Introduction} 
\label{sec:introduction}

Security protocols aim at providing secure communications on insecure networks by applying cryptographic primitives.
However, the design of security protocols is particularly error-prone.
Design flaws have been discovered for instance in the Google Single Sign On Protocol \cite{DBLP:conf/ccs/ArmandoCCCT08}, 5G \cite{DBLP:conf/ccs/BasinDHRSS18}, WiFi WPA2 \cite{DBLP:conf/ccs/VanhoefP17}, and TLS \cite{DBLP:conf/sp/CremersHSM16}.
These findings have made the verification of security protocols a very active research area since the 1990s.

During the last 30 years, many research efforts
\cite{DBLP:conf/cav/ArmandoBBCCCDHKMMORSTVV05,DBLP:conf/csfw/Blanchet01,DBLP:conf/fosad/EscobarMM07,DBLP:conf/cav/MeierSCB13,DBLP:conf/ccs/CremersHHSM17,DBLP:conf/ccs/FettKS16,DBLP:conf/sp/CremersHSM16,DBLP:conf/ccs/BasinDS15,DBLP:conf/esorics/CremersDM17}
were spent on designing techniques to model and analyze protocols.
The earliest protocol analysis tools, \eg,  the Interrogator \cite{DBLP:journals/tse/MillenCF87} and the NRL Protocol Analyzer \cite{DBLP:conf/csfw/Meadows96}, could be used to verify security properties specified in temporal logic. 
Generic model checking tools have been used to analyze protocols, \eg, FDR \cite{DBLP:conf/tacas/Lowe96} and later Murphi \cite{mitchell1997automated}. 
More recently the focus has been on model checking tools developed specifically for security protocol analysis, such as Blanchet's ProVerif \cite{DBLP:conf/csfw/Blanchet01}, the AVISPA tool \cite{DBLP:conf/cav/ArmandoBBCCCDHKMMORSTVV05}, Maude-NPA \cite{DBLP:conf/fosad/EscobarMM07} and tamarin prover \cite{DBLP:conf/cav/MeierSCB13}.
There have also been hand proofs aimed at particular protocols. 
Delaune \etal \cite{DBLP:journals/jcs/DelauneKS10} showed by a dedicated hand proof that for analyzing PKCS\#11 one may bind the message size. 
Similarly, Delaune \etal \cite{DBLP:conf/csfw/DelauneKRS11} gave a dedicated result for analyzing protocols based on the TPM and its registers.
Guttman~\cite{DBLP:journals/jar/Guttman12} also manually extended the space model by adding support for Wang's fair exchange protocol \cite{DBLP:journals/jcs/Wang06}.

Unfortunately, although formal analysis has been successful in finding design flaws in many security protocols,
it is still challenging for existing verification tools to support fully automated analysis of security protocols, especially protocols with global states \cite{PKCS11standard,yubikeymanual,DBLP:conf/crypto/GarayJM99,DBLP:conf/csfw/ArapinisRR11,DBLP:conf/ccs/Modersheim10} or unbounded sessions \cite{DBLP:journals/cacm/NeedhamS78,DBLP:conf/sosp/BurrowsAN89,DBLP:journals/sigops/OtwayR87}.
They may suffer non-termination during the verification mainly caused by the problem of state explosion.
To avoid the explosion of the state space, several tools, \eg, ProVerif \cite{DBLP:conf/csfw/Blanchet01} and AVISPA \cite{DBLP:conf/cav/ArmandoBBCCCDHKMMORSTVV05}, use an abstraction on protocols, so that they support more protocols with unbounded sessions.
Due to the abstraction, however, they may report false attacks when analyzing protocols with global states \cite{DBLP:conf/csfw/ArapinisRR11,DBLP:conf/csfw/BruniMNN15}.
StatVerif \cite{DBLP:conf/csfw/ArapinisRR11} and Set-$\pi$~\cite{DBLP:conf/csfw/BruniMNN15} extend the applied pi-calculus with global states, but the number of sessions they support is limited and they fail to automatically verify complicated protocols (\eg, CANauth protocol \cite{van2011canauth}).
GSVerif \cite{GSVerif-CSF18} enrichs ProVerif's proof strategy and supports several protocols with unbounded sessions \cite{DBLP:conf/crypto/GarayJM99,PKCS11standard}, but it fails to automatically verify complicated protocols (\eg, Yubikey protocol \cite{yubikeymanual}).
Tamarin prover \cite{DBLP:conf/cav/MeierSCB13,DBLP:conf/sp/KremerK14} can verify more protocols without limitations of states or sessions,
but it comes at the price of loosing automation.
It requires the user to supply insight into the problem by proving auxiliary lemmata, which is hard even for experts \cite{DBLP:conf/sp/KremerK14}.

We propose and implement \ourtool, a novel and general framework of verifying security protocols.
It pushes the limit of automation capability of state-of-the-art verification tools. 
Our work is motivated by the observation 
that these tools generally use a \textit{static} strategy during verification, where design of the strategy is non-trivial.
Here, the verification can be simply regarded as the process of path searching in a tree:  each node represents a proof state which includes a lemma as a candidate used to prove the lemma in its parent.
Due to the variation of states in searching spaces, the strategy must be general so that it can correctly choose the node in each searching step; otherwise it may result in non-termination.
It becomes more challenging when verifying different protocols, especially protocols with unlimited states and sessions.



Based on the observation, we design a \textit{dynamic} strategy in \ourtool. 
In other words, \ourtool runs round-by-round,
where in each round the strategy is either applied in searching until the complete proof path is selected,
or optimized in case the current selected path is estimated incorrect. 
The initialization of the strategy does not need any human intervention, \ie, the initial strategy is purely random. 
After the strategy is sufficiently optimized, it can smartly choose the next searching nodes. 
Specially, it efficiently localizes the node representing a \textit{supporting lemma} among the nodes, which leads to success in verification.  
Here, the supporting lemma is a special lemma necessarily used for proving the specified security property.
Therefore, the supporting lemmata have to be proven, before the property is verified.
Recall that tamarin prover can let users supply supporting lemmata to reduce the complexity of automation. 
In comparison, the dynamic strategy in \ourtool can help find the lemmata automatically and smartly, such that the protocols can be verified without any user interaction. 

Our dynamic strategy builds upon the insight that the node representing a supporting lemma is on the incorrect path with lower probability, when a random strategy is given (See the proof in Appendix \ref{sec:proof_of_our_insight}).
Hence, we introduce Deep Q Network~(DQN)~\cite{DBLP:journals/corr/MnihKSGAWR13}, a reinforcement learning agent, into the verification. 
The DQN updates the strategy according to historical incorrect paths. 
It uses an experience replay mechanism \cite{lin1993reinforcement} which randomly samples previous transitions, and smooths the training distribution over the incorrect paths.
The DQN supports the training of neural networks with stochastic gradient descent in a stable manner \cite{DBLP:journals/corr/MnihKSGAWR13}.
As a result, an optimized strategy tends to select a node representing supporting lemma among the candidates, which leads to higher probability of successful verification.

We also propose to solve other technical problems in implementing \ourtool.
We present an approach of generating incomplete verification tree for reducing the memory overhead.
We also design an algorithm of estimating correctness of selected paths, which is the key component for supporting the DQN.
Note that since we focus on the automation capability, we design \ourtool based on tamarin prover that we modify tamarin prover for preprocessing protocol models and acquiring information for the DQN.

Experimental results show that \ourtool can automatically verify all the studied protocols, without any human intervention.
These protocols include Yubikey protocol \cite{yubikeymanual} and CANauth protocol \cite{van2011canauth}, which cannot be automatically verified by state-of-the-art verification tools.
The case study also validates the efficiency of our dynamic strategy.

The main contributions of the paper are three folds:
\begin{enumerate}
        \item We present \ourtool, to the best of our knowledge, the first framework that automatically verifies security protocols by dynamic strategies.
        \item We propose several methods to deal with technical problems in implementing the framework.
        Specifically, we achieve our dynamic strategy by using the DQN.
        We present the design of rewards in DQN which corresponds to our insight.
        We design the algorithm of estimating the correctness of selected paths.
        We propose to generate the incomplete verification tree to reduce the memory overhead.
        We implement a multi-threading process of path-selecting for better efficiency.
        \item \ourtool pushes the limit of automation capability of protocol verification, and it greatly outperforms state-of-the-art tools. \ourtool achieves two goals: generality in designing heuristics and full automation in verification.
\end{enumerate}

The rest of the paper is organized as follows.
We review some related work and introduce tamarin that we use in Section~\ref{sec:relatedwork} and Section~\ref{sec:preliminaries}, respectively.
Then, we present the overview of \ourtool in Section~\ref{sec:framework}.
In Section~\ref{sec:example}, we show an illustrative example of security protocols.
Afterwards, we solve the main problems in designing the Acquisition and Verification module in Section~\ref{sec:constructtree} and Section~\ref{sec:dqn}, respectively.
We report our extensive experimental results and briefly overview the Yubikey protocol as case study in Section~\ref{sec:evaluation}.
Finally, we present our future work and conclude the paper.
We also illustrate and prove our insight in Appendix~\ref{sec:proof_of_our_insight}.
We present detailed description of the DQN in \Cref{sec:technique}.
Besides, we present an additional case study in \Cref{sec:casestudy}.

%% file: s-relatedwork.tex
\section{Related Work} 
\label{sec:relatedwork}

There are several typical model checking approaches that can deal with security protocols.
AVISPA/AVANTSSAR platform~\cite{DBLP:conf/cav/ArmandoBBCCCDHKMMORSTVV05} is a tool for the automated validation of security protocols and applications.
It provides a modular and expressive formal language for specifying protocols and their security properties, and integrates different backends that implement a variety of automatic analysis techniques.
It exhibits a high level of scope and robustness while achieving performance and scalability.
Moreover, Fr{\"{o}}schle \etal \cite{DBLP:conf/ifip1-7/FroschleS09} automatically analyze APIs of key tokens using SATMC of AVISPA~\cite{DBLP:journals/ijisec/ArmandoC08}, and show that an API is not robust by finding an attack.

ProVerif \cite{DBLP:conf/csfw/Blanchet01}, one of the most efficient and widely used protocol analysis tools, relies on an abstraction that encodes protocols in Horn clauses. 
This abstraction is well suited for the monotonic knowledge of an attacker, which makes the tool efficient for verifying protocols with an unbounded number of protocol sessions \cite{BhargavanBlanchetKobeissiSP2017,BlanchetCSF17,KobeissiBhargavanBlanchetEuroSP17}. 
It is capable of proving reachability properties, correspondence assertions, and observational equivalence.
Protocol analysis is considered with respect to an unbounded number of sessions and an unbounded message space.
StatVerif~\cite{DBLP:conf/csfw/ArapinisRR11} is an extension of ProVerif with support for explicit states.
Its extension is carefully engineered to avoid many false attacks.
It is used to automatically reason about protocols that manipulate global states.
GSVerif~\cite{GSVerif-CSF18} extends ProVerif to global states.
It provides several sound transformations that cover private channels, cells, counters, and tables.
It is efficient to verify protocols with global states.
However, it still fails to verify several complicated protocols (\eg, Yubikey \cite{yubikeymanual} and CANauth protocol \cite{van2011canauth}) with unbounded sessions.
For these protocols, it requires users to manually design a sequence of proof formulas to prove the security properties.



Another verification approach that supports the verification of stateful protocols is the tamarin prover \cite{DBLP:conf/csfw/SchmidtMCB12}, \cite{DBLP:conf/cav/MeierSCB13}. 
Instead of abstraction techniques, it uses backward search and lemmata to cope with the infinite state spaces in verification. 
The benefit of tamarin and related tools is a great amount of flexibility in formalizing relationships between data that cannot be captured by a particular abstraction and resolution approach. 
It can handle protocols with global states \cite{DBLP:conf/sp/KremerK14}, unbounded sessions \cite{DBLP:conf/cav/MeierSCB13}, observational equivalence properties \cite{DBLP:conf/ccs/BasinDS15} and XOR \cite{DBLP:conf/ccs/BasinDHRSS18} \etc
However it comes at the price of loosing automation, \ie, the user has to supply insight into the problem by proving auxiliary lemmata for complex protocols.
Tamarin has already been used for analyzing the Yubikey device \cite{DBLP:conf/stm/KunnemannS12}, security APIs in PKCS\#11 \cite{DBLP:journals/jcs/DelauneKS10} and a protocol in TPM \cite{DBLP:conf/csfw/DelauneKRS11}.
Moreover, it does not require bounding the number of keys, security devices, reboots, \etc

As a total, current approaches provide efficient ways in verifying security protocols.
However, they commonly adopt a static strategy during verification, which may result in non-termination when verifying complicated security protocols.
Encountering these cases, human experts are needed to analyze the reason of non-termination and supply hand proof.

At the same time, fast progress has been unfolding in machine learning applied to tasks that involve logical inference, such as natural language question answering \cite{DBLP:conf/nips/SukhbaatarSWF15}, knowledge base completion \cite{DBLP:conf/nips/SocherCMN13} and premise selection in the context of theorem proving \cite{DBLP:conf/nips/IrvingSAECU16}. 
Reinforcement learning in particular has proven to be a powerful tool for embedding semantic meaning and logical relationships into geometric spaces, specifically via models such as convolutional neural networks, recurrent neural networks, and tree-recursive neural networks. 
These advances strongly suggest that reinforcement learning may have become mature to yield significant advances in many research areas, such as automated theorem proving.
To the best of our knowledge, \ourtool is the first work that applies AI techniques to the automated verification of security protocols.

%% file: s-preliminaries.tex
\section{Preliminaries of Tamarin Prover} 
\label{sec:preliminaries}


We firstly introduce tamarin prover that we modify. 
The tamarin prover \cite{DBLP:conf/cav/MeierSCB13} is a powerful tool for the symbolic modeling and analysis of security protocols.
It takes a protocol model as input, specifying the actions taken by protocol's participants (\eg, the protocol initiator, the responder, and the trusted key server), a specification of the adversary, and a specification of the protocol's desired properties.
Tamarin can then be used to automatically construct a proof that, when many instances of the protocol's participants are interleaved in parallel, together with the actions of the adversary, the protocol fulfills its specified properties. 

Protocols and adversaries are specified using an expressive language based on multiset rewriting rules.
These rules define a labeled transition system whose state consists of a symbolic representation of the adversary's knowledge, the messages on the network, information about freshly generated values, and the protocol's state.
The adversary and the protocol interact by updating and generating network messages.
Security properties are modeled as trace properties, checked against the traces of the transition system.

To verify a protocol, tamarin uses a constraint solving algorithm for determining whether $P \vDash _{E} \phi$ holds for a protocol~$P$, a trace property $\phi$
 and an equational theory $E$ that formalizes the semantics of function symbols in protocol model.
The verification always starts with either a simplification step, which looks for a counterexample to the property, or an induction step, which generates the necessary constraints to prove the property.
Constraint systems $\Upgamma$ are used to represent the intermediate states of search.
This problem is undecidable and the algorithm does not always terminate.
Nevertheless, it often finds a counterexample (an attack) or succeeds in unbounded verification.


On a high level, a constraint solving algorithm consists of two components: (1) a constraint reduction strategy, which gives rise to a possibly infinite search tree, and (2) a search strategy, which is used to search this tree for a solved constraint system.

A constraint reduction strategy is a function $r$ from constraint systems $\Upgamma$ to a finite set of constraint systems satisfying the following two conditions.
(i) The constraint reduction relation $\{(\Upgamma, r(\Upgamma)) \mid \Upgamma \in dom(r)\}$ is well-formed, correct, and complete.
(ii) Every well-formed constraint system not in the domain of $r$ has a non-empty set of solutions.

Intuitively, tamarin applies constraint reduction rules to a constraint system to generate a finite set of refined systems.
Note that tamarin prover uses these rules to represent lemmata in verification process.
The rule application process is repeated until the constraint system is solved.
Therefore, tamarin uses a search tree to store all the constraints systems in the verification process.
Each node of a tree is an independent constraint system.

\begin{figure}[t]
\centering
\includegraphics[width=8cm]{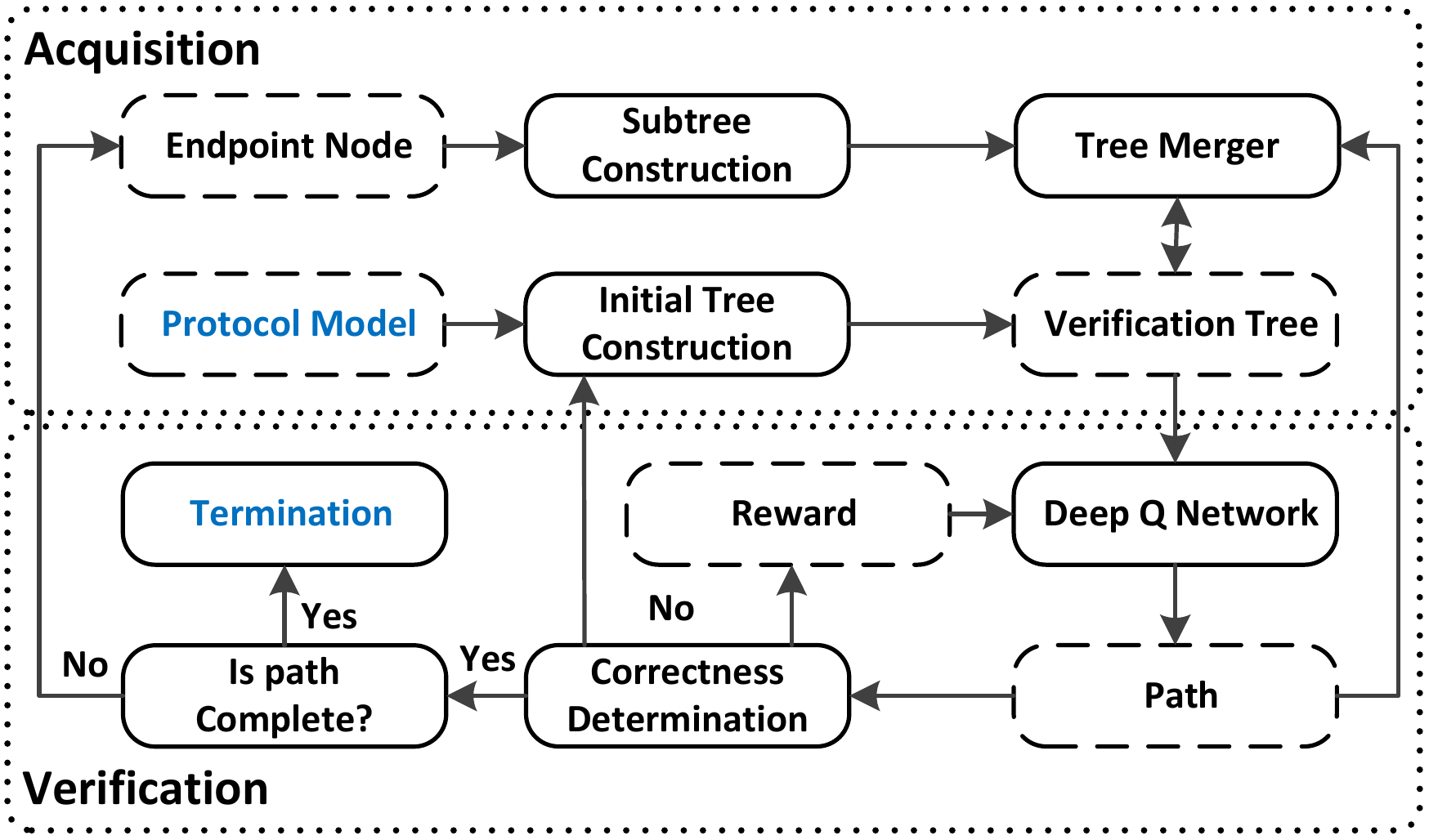}
\vspace{-0.1in}
\caption{Framework of \ourtool.}
\label{fig:framework}
\vspace{-0.2in}
\end{figure}

To search for a solved constraint system, tamarin prover uses a heuristic to sort the applicable rules for a constraint system.
The design rationale underlying tamarin's heuristic is that it prefers rules that are either trivial to solve or likely to result in a contradiction.
Since the rules are sorted, tamarin always chooses the first rule to apply to the constraint system.
Hence, the search tree is simplified to a finite one, which reduces the complexity of verification process.


Tamarin provides two ways of verifying protocols.
It has an efficient automated mode that combines message deduction and human-designed heuristics to guide proof search. If the automated proof search terminates, it returns either a proof of correctness or a counterexample, representing an attack that violates the stated property. 
Since the tool may not terminate in automated mode, tamarin also provides an interactive mode for users to manually construct a proof and seamlessly combine it with automated proof search.

%% file: s-framework2-hwc.tex
\section{Overview} 
\label{sec:framework}




Briefly, given a tamarin protocol model,~\ie, a protocol description and a security property, the workflow of \ourtool consists of the following steps:

\begin{itemize}
        \item Step 1: The Deep Q Network (DQN) is initialized with a purely random strategy, which takes multiple candidates as input, and randomly chooses a candidate with the uniform probability as output.

        \item Step 2: \ourtool conducts proof searching by using the current strategy. It executes in parallel by multi-threading. In each thread, a proof path is generated using the tamarin prover as backend, where each node on the path is chosen according to the strategy.

        \item Step 3: If a path generated in Step 2 is correct and complete, \ourtool terminates and outputs the path as the result.

        \item Step 4: Otherwise, if all the paths generated in Step 2 are estimated incorrect, \ourtool starts a new epoch where the DQN is trained according to the proof paths, and the strategy is updated. Here, we use the term epoch to denote the time step in which the DQN is optimized with new rewards.

        \item Step 5: Go to Step 2.
\end{itemize}

As shown in \Cref{fig:framework}, \ourtool contains Acquisition and Verification module, which execute in multiple rounds:

 \textbf{Acquisition}: The module generates a verification tree as input of Verification module.
A path in the tree corresponds to a possible proof path in verification.
Each node in the tree contains information for guiding the verification.
We modify tamarin prover to collect the information.
Since information of the nodes is transformed as input of the DQN, which is difficult for handling high dimensional data,
 the information must be carefully chosen to reduce complexity of designing the DQN.
Moreover, we face a problem of constructing the verification tree.
There are protocols with large or infinite state spaces \cite{PKCS11standard,yubikeymanual}.
In this case, even little information stored in nodes would still lead to memory explosion.
To solve the problem, we design the DQN to guide the tree generation.
Specifically, the tree is generated and expanded gradually that in each round only one of the endpoint nodes in the current tree is expanded, and the rest endpoint nodes remain collapsed for reducing the complexity of the tree. Here, the selection of the endpoint node is guided by the current strategy in DQN.



 \textbf{Verification}: The module selects a path from the verification tree as a candidate proof.
The path selection is guided by a dynamic strategy which uses the DQN.
Meanwhile, the strategy is also optimized with correctness of the selection. 
Here, we additionally illustrate how the submodules in the Verification module deal with the tree.
Briefly, there are 2 submodules. 

\textbf{1) Correctness Determination:} It estimates whether the DQN selects the correct path. 
The main idea is to detect whether there are loops along the path (See \Cref{subsec:correctness}).
In each round, \ourtool works according to the selected path in different cases:

\textbf{Case 1:} 
\textit{The path is estimated incorrect}.
We optimize the DQN in this epoch by passing rewards to the DQN.
Meanwhile, we start a new epoch and send feedback to Acquisition module, where the submodule of Initial Tree Construction is 
informed to regenerate a new verification tree.
As a result, we can find a new proof path according to the optimized DQN afterward.

\textbf{Case 2:} \textit{The path is estimated correct but incomplete}.
The incompleteness of the path is caused by the incompleteness of the verification tree. 
Therefore, we inform the submodule of Subtree Construction to expand the tree in the next round, so the path is also extended in the next round for shaping a complete path.

\textbf{Case 3:} \textit{The path is correct and complete}. 
In this case, we achieve a successful verification of the protocol model, so we can terminate \ourtool.

\textbf{2) Deep Q Network:}
We introduce the DQN to update the dynamic strategy in \ourtool.
The key of the design of DQN is constructing the reward.
Specifically, for each node that is on an estimated incorrect path, the node is bound to a negative reward.
The design of the reward corresponds to our insight as mentioned in \Cref{sec:introduction}.
This insight enables us to leverage the detected paths to guide the path selection.

%% file: s-example.tex
\section{Example} 
\label{sec:example}

To illustrate our method, we consider a simple security protocol.
The goal of the protocol is that when a participant $C$ sends a symmetric key $k$ to another participant $S$, the secret symmetric key $k$ should not be obtained by the adversary.

\begin{figure}[h]
  \vspace{-0.2in}
\begin{align*}
& S_1.\ C \rightarrow S: && \{{T}_1,\ k\}_{{pk}_s} \\
& S_2.\ S \rightarrow C: && \{{T}_2,\ h(k)\} \\
\end{align*}
  \vspace{-0.4in}
\caption{A simple security protocol.}
\label{fig:protocol}
  \vspace{-0.1in}
\end{figure}

The brief process of the protocol is shown in \Cref{fig:protocol}.
In step~$S_1$, $C$ generates a symmetric key $k$, encrypts a tuple~$\{{T}_1, k\}$ with the public key of $S$, and sends the encrypted message.
Here, tag ${T}_i$ is used to annotate protocol step $i$ in protocol execution. 
In step~$S_2$, $S$ receives $C$'s message, decrypts it with its private key, and gets the symmetric key $k$.
Finally, $S$ confirms the receipt of $k$ by sending back its hash $h(k)$ to $C$.

The communication network is assumed to be completely controlled by an active Dolev-Yao style adversary \cite{DBLP:journals/tit/DolevY83}. In particular, the adversary may eavesdrop the public channels or send forged messages to participants according to the channels. 
Moreover, the adversary can access the long-term keys of compromised agents.
Besides, the adversary is limited by the constraints of the cryptographic methods used. For example, it cannot infer hash input from hash output.

Here, we provide a brief explanation on modeling protocols in tamarin prover.
A tamarin model defines a transition system whose state is a multiset of facts.
The transitions are specified by rules.
At a very high level, tamarin rules encode the behavior of participants and adversaries.
Tamarin rules~$l - [a] \rightarrow r$ have a left-hand side $l$ (premises), actions~$a$, and a right-hand side $r$ (conclusions).
The left-hand and right-hand sides of rules respectively contain multisets of facts.
Facts can be consumed (when occurring in premises) and produced (when occurring in conclusions).
Each fact can be either linear or persistent (marked with an exclamation point~$!$).
While we use linear facts to model limited resources that cannot be consumed more times than they are produced, persistent facts are used to model resources which can be consumed any number of times once they have been produced.
Actions are a special kind of facts.
They do not influence the transitions, but represent specific states in protocol.
These states form the relation between transition system and the security property.

Security properties are specified in a fragment of first-order logic.
Tamarin offers the usual connectives (where $\&$ and $|$ denote ``and'' and ``or'', respectively), quantifiers All and Ex, and timepoint ordering~$<$. In formulas, the prefix $\#$ denotes that the following variable is of type timepoint.
Besides, tamarin offers two connectives $@$ and $\triangleright_\textup{o}$ for stating the relations between facts and timepoints.
For example, the expression $\kwl{Action}(args)@\#t$ denotes that $\kwl{Action}(args)$ is executed at timepoint $\#t$.
The expression $\kwl{Action}(args)\triangleright_\textup{o}\#t$ denotes that $\kwl{Action}(args)$ is executed before timepoint $\#t$.

For instance, to model the above protocol, we first define several functions and predicates.
1) $\kwl{In}(m)$ and $\kwl{Out}(m)$: message $m$ is sent and received, respectively;
~2)~$\kwl{aenc}\{a\}k$ and $\kwl{adec}\{a\}k$: asymmetric encryption and decryption of a variable $a$ using key $k$;
3) $\kwl{Pk}(A,\ {pk}_A)$ and $\kwl{Ltk}(A,\ {ltk}_A)$: participant $A$ is bound to a public key ${pk}_A$ and a private key ${ltk}_A$, respectively;
4) $\kwl{fst}\{a,\ b\}$ and $\kwl{snd}\{a,\ b\}$: the first and second element from a tuple $\{a,\ b\}$, respectively;
5) $\kwl{Eq}(a,\ b)$: $a$ is equal to $b$;
6) $\kwl{h}(a)$: the result of hashing $a$.


Then, the compromise of private keys is modeled using the following
rule. 

  \vspace{-0.2in}
\begin{align*} 
    & \kwl{rule}\ Reveal\_ltk:  \\
    & [\ \kwl{!Ltk}(A,\ {ltk}_A)\ ] - [\ \kwl{LtkReveal}(A) \ ]  \rightarrow   [\ \kwl{Out}({ltk}_A)\ ]   
\end{align*} 

It has a premise $\kwl{!Ltk}(A,\ {ltk}_A)$ which binds the private key ${ltk}_A$ to a participant $A$.
The corresponding conclusion~$\kwl{Out}({ltk}_A)$ states that the private key ${ltk}_A$ is sent to the adversary.
Note that, this rule has an action $\kwl{LtkReveal}(A)$ stating that the key of $A$ was compromised.
This action is used to model the security property.

\begin{figure}[tb]
\resizebox{.9\linewidth}{!}{
\begin{minipage}{\linewidth}
  \begin{align*} 
    & \kwl{rule}\ C\_1 : \\
    & [\ \kwl{Fr}(k),\ \kwl{!Pk}(S,\ {pk}_S)\ ]  \rightarrow  [\ \kwl{Send}(S,\ k), \ \kwl{Out}(\kwl{aenc}\{{T}_1,\  k\}{pk}_S)\ ] \ \\ 
    & \\
    &\kwl{rule}\ C\_2:   \\
    &[\ \kwl{Send}(S,\ k),\ \kwl{In}(\kwl{h}(k))\ ]   -[\ \kwl{SessKeyC}(S,\ k)\ ]\rightarrow [] \\
    & \\
    & \kwl{rule}\ S\_1:   \\
    & [\ \kwl{!Ltk}(S,\ {ltk}_S),    \ \kwl{In}(request)\ ]    -[\ \kwl{Eq}(\kwl{fst}(\kwl{adec}(request,\ {ltk}_S)),\ {T}_1)\ ] \\ 
    & \rightarrow  [\  \kwl{Out}({T}_2,\kwl{h}(\kwl{snd}(\kwl{adec}(request,\ {ltk}_S))))\ ]  \\ 
  \end{align*} 
\end{minipage}
}
  \vspace{-0.1in}
  \caption{The model of the protocol process.}
  \label{fig:rule}
\vspace{-0.2in}
\end{figure}

Then, the protocol is modeled using the rules in \Cref{fig:rule}.
Rule~\textit{C\_1} captures a participant generating a fresh key and sending the encrypted message.
The rule has two facts for premises.
The first fact $\kwl{Fr}(k)$ states that a fresh variable $k$ is generated.
The second fact $!\kwl{Pk}(S,\ {pk}_S)$  states that the public key ${pk}_S$ is bound to a participant~$S$.
In this case, the second fact is a persistent fact since the public key can be used in many times (\ie, by protocol participants or adversaries).
If the facts in premises are matched with the facts in the current state, two conclusions are produced.
The first is an action $\kwl{Send}(S,\ k)$ which states that $k$ is sent to a participant~$S$.
The second conclusion is~$\kwl{Out}(\kwl{aenc}\{{T}_1,\  k\}{pk}_S)$.
This fact states that the participant uses a public key ${pk}_S$ to encrypt the message $\{{T}_1,\  k\}$ and send the message.
Rule~\textit{S\_1} captures a participant receiving the message sent by $C$ and sending the hash value of $k$ back.
Rule~\textit{C\_2} captures a participant receiving the hash value and completing a run of the protocol.

Finally, we define a security property,
which states that when a participant $C$ sends a symmetric key $k$ to another participant $S$ , the secret symmetric key $k$ should not be obtained by the adversary.
The security property is modeled as a lemma \textrm{Key\_secrecy} in \Cref{fig:property}.
The lemma indicates that,
there must not exist a state,
where action $\kwl{SessKeyC}(S,\ k)$ happens and the adversary obtains $k$,
without the happening of the compromise action~$\kwl{LtkReveal}(S)$.

\begin{figure}[ht]
  \vspace{-0.2in}
  \begin{align*} 
    & \kwl{lemma}\ Key\_secrecy:   \\
    &  \ \ '' not( \     Ex\ S\ k\ \#i\ \#j.  \       SessKeyC(S,\ k)\ @\ \#i \ \&\ K(k)\ @\ \#j \\
    &  \ \ \ \ \   \&\ not(Ex\ \#r.\ LtkReveal(S)\ @\ r) \   )'' \\
  \end{align*} 
  \vspace{-0.5in}
  \caption{The security property.}
  \label{fig:property}
\vspace{-0.1in}
\end{figure}

Note that the above security property of protocol can be successfully verified by tamarin prover.
To better understand the following sections, we use this protocol as an example.
We describe how we generate the verification tree of the protocol in \Cref{sec:constructtree}.
Then we explain how \ourtool verifies protocols in \Cref{sec:dqn}.

%% file: s-constructtree.tex
\section{Acquisition module}
\label{sec:constructtree}

\subsection{Choosing Information}
\label{subsec:tamarin}

\begin{figure*}[tbp]
\centering
\includegraphics[width=16cm,height=4.5cm]{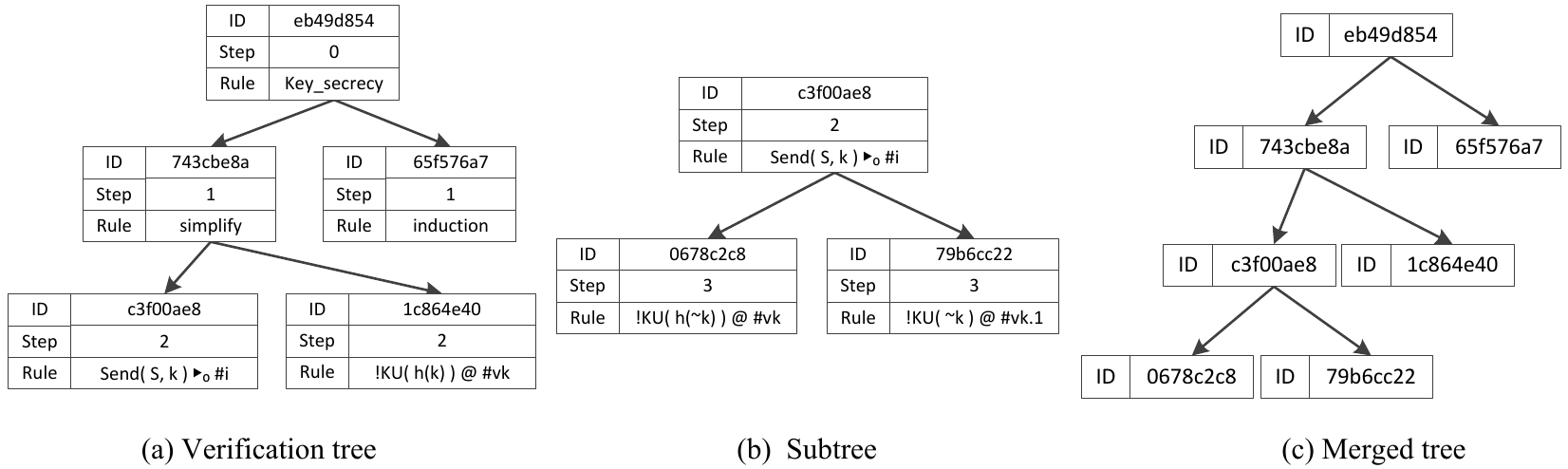}
\vspace{-0.1in}
\caption{Construction of a Verification Tree.}
\label{fig:tree}
\vspace{-0.1in}
\end{figure*}

The information in nodes of the verification tree is used in 2 ways.
1) We transform the information to input of the DQN.
In the Verification module, we use the DQN to select a proof path in verification tree.
The DQN in Verification module requires an input state,
which represents current proof state.
We use the information to represent proof state in verification process.
Since it is difficult for the network to handle high-dimensional data,
the input of the network should not be large in dimensions.
Hence, we do not choose all the intermediate data in the verification process as the information.
2) We use the information to distinguish different proof states.
Note that \ourtool runs round-by-round,
where in each round verification trees are constructed and merged.
In the merging process, we compare the information in nodes in different trees to find a same proof state in each round.
Therefore, the information in the node should not only be simple enough, but also represent independent state in the verification process.

For each node in the verification tree, we choose the constraint reduction rule and and step number, \ie, distance from the node to the root, as the stored information.
Recall that at each proof step tamarin prover applies a constraint reduction rule to refine a constraint system.
Hence, rules and their step numbers can represent independent states in the verification process.
As illustrated in \Cref{fig:tree}, each node in the tree contains three pieces of information as follows:
1) ID: the hash value of the constraint reduction rule;
2) Step: the current proof step number;
3) Rule: the string of the constraint reduction rule.
Here, the hash value is shown with the first eight characters for abbreviation.

Considering the protocol in \Cref{sec:example}, 
we show the information collected from modified tamarin prover in \Cref{fig:tree}~(a).
Due to the limitation of paper size, we only demonstrate the information collected in the first two steps.
Specifically, the root node $eb49d854$ represents the lemma~$Key\_secrecy$.
In the first proof step, tamarin prover constructs the proof starting with either an~$induction$ rule, which generates the necessary constraints to prove the lemma, or 
a~$simplify$ rule, which generates initial constraint system to look for a counterexample to the lemma. 
In this case, tamarin prover chooses rule~$simplify$ at proof step~\#1, which corresponds to node $743cbe8a$.
In detail, it looks for a protocol execution that contains a~$\kwl{SessKeyC}(S,\ k)$ and a~$\kwl{K}(k)$ action, but does not use an~$\kwl{LtkReveal}(S)$.
As shown in \Cref{fig:rule}, action~$\kwl{SessKeyC}(S,\ k)$ is in the protocol execution only if the protocol step that rule~$C\_2$ captures has happened.
Since rule $C\_2$ has two premises~$\kwl{Send}(S,\ k)$ and~$\kwl{In}(\kwl{h}(k))$,
these two facts are in the protocol execution.
Based on this observation, tamarin prover has two constraint reduction rules to select: $Send( S,\ k ) \triangleright_\textup{o}\#i$ and~$KU( h(k) )@ \#vk$.
The first rule $Send( S,~k )\triangleright_\textup{o}\#i$ states that action~$\kwl{Send}(S,\ k)$ is executed before timepoint $\#i$.
The second rule $KU( h(k) )@ \#vk$ states that the adversary knows $k$'s hash value at timepoint $\#vk$.
Therefore, in the tree, node $743cbe8a$ has two children which represent these two rules respectively.


\subsection{Tree Construction}
\label{subsec:treeconstruction}

We construct a verification tree to store the information we collect.
The tree is used in Verification module to generate a candidate proof by guidance of the DQN. 
As illustrated in~\Cref{sec:framework}, to avoid memory explosion, we design a simple and effective approach.
We firstly initialize a tree with a root starting from the security property.
Each node in the tree contains information specified in \Cref{subsec:tamarin}.
Then, in each new round, when the tree is expanded, an endpoint node in the current tree is chosen according to the DQN, and a two-depth subtree is generated.
The root of the subtree is the chosen endpoint node and the nodes of the second depth represent the possible constraint reduction rules that can be used to prove the lemma of the root.
Therefore, the new tree is formed by merging the subtree into the current tree.

In \Cref{fig:tree}, we exemplify the construction of the verification tree for the protocol in \Cref{sec:example}.
In the initial round, the Acquisition module generates a tree, whose root node represents the lemma~$Key\_secrecy$, as shown in \Cref{fig:tree}(a).
In the next round, if the DQN in the Verification module selects the endpoint node $c3f00ae8$ as the estimated supporting lemma, the Acquisition module uses the modified tamarin prover to gather the information, and 
constructs a two-depth subtree shown in \Cref{fig:tree}(b).
Then, the acquisition merges the subtree into the current tree as shown in \Cref{fig:tree}(c).

Besides, we implement a multi-threading process of path-selecting for better efficiency.
Recall that the DQN optimizes itself with its selection and the corresponding rewards.
Since it selects only one path given a verification tree at each round, the quantity of training data is limited, which decreases the training efficiency and lowers the performance.
To solve this problem, we execute multiple threads of the Acquisition module in parallel to generate various verification trees for the DQN.
Therefore, the DQN selects multiple paths in these trees at a time and generate more training data for optimizing.
Using this approach, we are able to achieve greater data efficiency and increase the convergence rate.
We further validate and evaluate the multi-threading process in the experiments in \Cref{subsec:mainevaluation}.

%% file: s-dqn-sc.tex
\section{Verification module}
\label{sec:dqn}

Briefly, the Verification module selects a proof path from the verification tree.
The selection is guided by our dynamic strategy and an algorithm of correctness determination of the selection.
In \Cref{subsec:correctness}, we describe our method of correctness determination.
We describe the design of DQN in \Cref{subsec:dqn}.
We then analyze the DQN in \Cref{subsec:strategy}.


\subsection{Correctness Determination}
\label{subsec:correctness}

We illustrate how we determine the correctness of a proof path.  
The main idea is to detect whether there are loops along the path.
For example,  \Cref{fig:infinite-loop} shows the information of the last 5 consecutive nodes on a path, when using tamarin and encountering the loop on the path. 
Here, words in blue indicate constraint reduction rules selected by tamarin prover, and words in black indicate tags marked by tamarin prover in verification.
We find that the rules of the nodes are similar, which has the following form:
\begin{align}
!\kwe{KU(aenc(< '2', ~ni.1, nr.1, \$R.1>, pk(~ltkA.1)))\ @\ \#vk.1) } \nonumber
\end{align}
Therefore, we consider the path incorrect since the loop on it results in a non-termination in searching.
Besides, there are several other kinds of loops.
For example, the sequence of the nodes may have the form $[\dots, a, b, a, b]$ or $[\dots, a, b, c, a, b, c]$, where $a, b, c$ are constraint reduction rules.
The loops on these sequences may also lead to a failed verification.
Hence, the loop detection algorithm should be carefully designed to determine the correctness of the path.

\begin{figure}[tbp]
\centering
\includegraphics[scale=0.25]{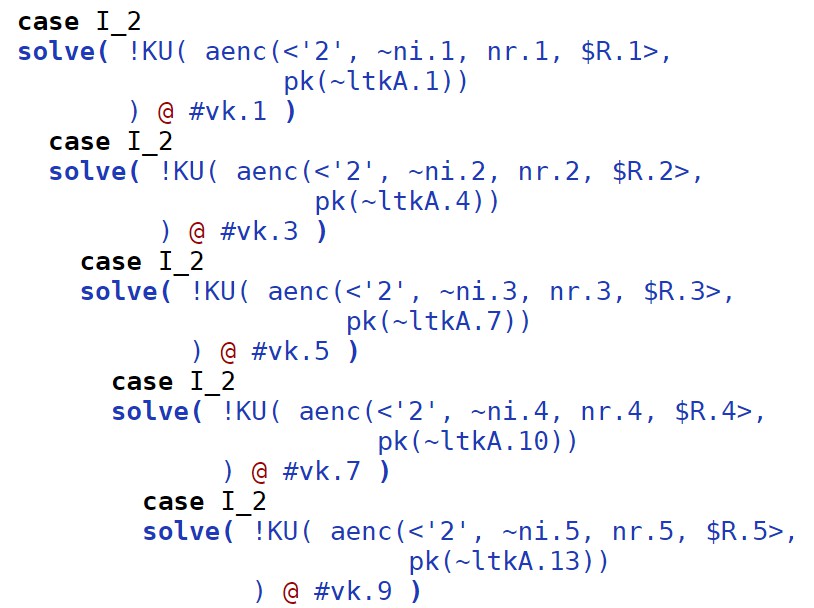}
\caption{An example of verification loop.}
\label{fig:infinite-loop}
\end{figure}

Based on the observation, we design the algorithm of loop detection as shown in Algorithm 1.
The algorithm takes a string sequence $[s_{1}, s_{2},..., s_{k}]$ as input.
The sequence is transformed from the selected path.
Each element, \eg, $s_i$, is the constraint reduction rule of the corresponding node, \ie, the $i$th node, on the path.
The algorithm runs iteratively by generating a sequence~$S$ (line~4), and counting the number of pairs of similar elements that are both in $S$. 
If the number is not less than~$\delta$, a loop is detected.
For example, given the aforementioned sequence $[\dots, a, b, c, a, b, c]$,  $S=(c,c, \dots)$, when $j=3$. 
Note that the length of $S$ is restricted to the limit of $\rho$ for efficiency.
For the same reason, the algorithm is activated only when the length of the selected path reaches $\alpha$ (line 1).



We use levenshtein distance to measure the similarity between two rules (line 8 to 13).
The distance between two strings $x$ and $y$ is the minimum number of single-character edits (insertions, deletions or substitutions) required to change $x$ into $y$.
If the distance between $x$ and $y$ is less than $\beta$ of the length of $x$, we assume these two strings are similar.
For example in \Cref{fig:infinite-loop}, the difference points among the similar rules are often the numbers at corresponding positions, \eg, $ni.1, ni.2, ni.3$.
Informally, the number of difference points are counted as the distance.

Note that there is not necessarily a loop on every incorrect path. 
In other words, if a loop is found given a selected path, there may be multiple \textbf{incorrect nodes}, \ie, the nodes that do not represent supporting lemmata, on the path.
Therefore, it is not sufficient to use naive search algorithms, \eg, DFS, to locate proof paths.
We make further studies on analyzing the effectiveness of the algorithms in \Cref{sec:evaluation}.
Finally, in our implementation, we set $\alpha$ to $20$, $\beta$ to $0.1$, $\rho$ to $20$ and $\delta$ to $3$.

\begin{algorithm}[t]
\caption{Loop Detection Algorithm}  
\label{alg:similarity}  
\algrenewcommand{\ALG@beginalgorithmic}{\footnotesize}
\begin{algorithmic}[1] 
\Require $(s_{1}, s_{2},..., s_{k})$ 
\If{$k \geq \alpha$}
\For{$j=1$ to $k-1$}
\State $count = 0$
\State $S = (s_{k}, s_{k-j}, s_{k-2j},...) \ \  $//$ |S| \le \rho$
\For{each element $x$ in $S$}
\For{each element $y$ in $S$}
\If{$x \neq y$}
\State $n$ = length of $x$, $m$ = length of $y$
\For{$a=1$ to $n$}
\For{$b=1$ to $m$}
\State \textbf{if} $y[a]==x[b]$ \textbf{then} $cost=0$ \textbf{else} $cost=1$
\State  $v[b] = \min{(v[b-1] + 1, b, b-2 + cost)}$
\EndFor
\EndFor
\State \textbf{if} $v[m] / m$ $< \beta$ \textbf{then} $count = count + 1$ 
\EndIf
\EndFor
\EndFor
\If{$count$ $\geq \delta$ }
\Return TRUE 
\EndIf
\EndFor
\State \textbf{return} FALSE 
\EndIf
\State \textbf{else} \textbf{return} FALSE 
\end{algorithmic}  
\end{algorithm}

\subsection{Deep Q Network}
\label{subsec:dqn}

To apply our insight to \ourtool, we introduce Deep Q Network (DQN) into the verification.
Specifically, in each epoch, the DQN uses a dynamic strategy to select paths in the verification tree.
When the paths are estimated incorrect, the DQN optimizes its current strategy by training.
Given a group of estimated incorrect paths, the training process works as follows.
\textit{First}, for each node that is on a path of the group, the node is bound to a negative reward. 
The reward is related to the probability that the node represents a supporting lemma, according to our insight.
In the implementation, a tuple, which includes the node and the reward, is added into a global dataset.
\textit{Then}, a subset is sampled from the dataset.
\textit{Finally}, the parameters of the DQN is optimized by minimizing a loss function.
Here, the loss function is the sum of sub-functions.
Each sub-function takes an individual tuple in the subset as input, and outputs the difference between results calculated according to two pre-defined functions.
One of the functions is calculated by using the parameters of current DQN, and the other is calculated by using the training parameters of the optimized DQN.
Informally, the training process optimizes the network in order that the reward of each node can be estimated. 
In other words, the node with the highest reward among its siblings can be regarded as the one with supporting lemma, if the DQN is sufficiently optimized.
We demonstrate the technical details of our implementation of DQN in \Cref{sec:technique}.

\subsection{Analysis of our DQN}
\label{subsec:strategy}

\input{s-strategy}

%% file: s-strategy.tex

The advantage of applying DQN is that DQN can update our dynamic strategy efficiently if the rewards in DQN are designed effectively.
In other words, sufficient tuples are required to be updated in the dataset in each epoch, and the reward in each tuple should be correct.
A naive algorithm is that only the reward corresponding to the last node on the incorrect path is set negative.   
However, the number of updated tuples is limited.
Instead, we significantly improve the efficiency by storing multiple tuples as illustrated, and the correctness of the insight is analyzed as follows.

Recall that the reward in our DQN is related to the probability that the node represents a supporting lemma.
Formally, suppose there are $R$ correct and complete proof paths in a given verification tree, denoted as $[n_{t_1}, n_{t_2}, ..., n_{t_{k_t}}]$, where $n_{t_i}$ is the $i$th proof state at the $t$th path.
Therefore, the lemmata in $\{n_{t_i}\}$ are the candidate lemmata. 
If $n_{t_i}$ has $x$ children, and $y$ of them represent supporting lemma, and the random strategy is applied in choosing child, then the probability of choosing a node representing supporting lemma is $\frac{y}{x}$.
The random strategy here means that whenever choosing a child for searching, the probability of choosing is uniform.
Suppose there are at least one child of $n_{t_i}$ that does not represent supporting lemma, \ie, $x>y$. 
As demonstrated in the following theorem, a node with supporting lemma is on the incorrect path with lower probability, when a random strategy is given.
Therefore, the rewards of nodes on estimated incorrect paths are set negative.

\begin{theorem}
Given the above assumptions, after $n_{t_i}$ has been chosen, 
the node representing a supporting lemma, who is the child of $n_{t_i}$, is on an incorrect path with the probability less than $\frac{y}{x}$.
\end{theorem}
\begin{proof}
        See Appendix \ref{sec:proof_of_our_insight}.
\end{proof}

\begin{figure}[tbp]
\centering
\includegraphics[width=8cm,height=4.5cm]{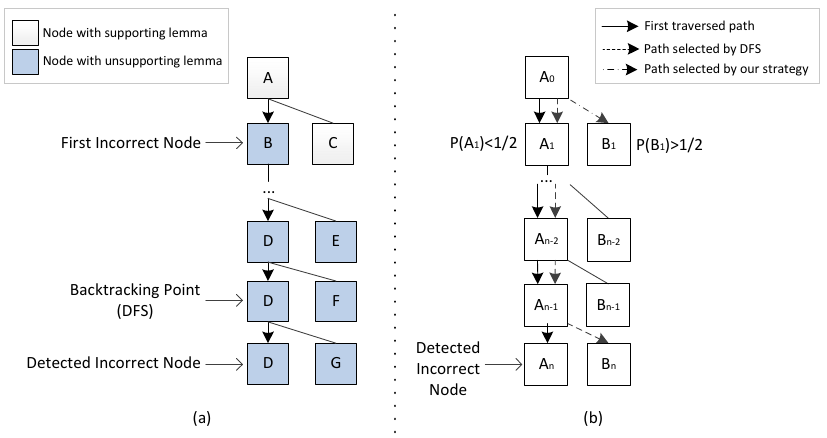}
\begin{center}
\vspace{-0.2in}
\caption{Verification trees: (a) A non-zero distance between the first incorrect node and the detected incorrect node. (b) The probability of nodes representing supporting lemmata after the first traversing.}
\label{fig:dqn}
\end{center}
\vspace{-0.4in}
\end{figure}

To illustrate the theorem, we briefly study a naive and seemly good algorithm, which
 traverses the verification tree by DFS, selects nodes using existing static strategy when traversing, and backtracks according to our algorithm of correctness determination.

\textbf{1) Insufficiency of only applying correctness determination:} The naive algorithm may have to explore a large amount of nodes. 
For example in \Cref{fig:dqn} (a), since the algorithm of correctness determination cannot detect all incorrect nodes, there is a non-zero distance between the first incorrect node $B$ and the first detected incorrect node $D$.
Hence, the naive algorithm has to traverse all paths that pass through $B$, before it traverses the correct path.
As a result, the number of explored nodes grows exponentially with the distance.


\textbf{2) Advantages of our DQN:} Our DQN greatly outperforms the naive algorithm according to the theorem.
For example in \Cref{fig:dqn} (b), assume that our DQN and the naive algorithm traverses along the same  path [${{A}}_0,{{A}}_1,...,{{A}}_n$], until the algorithm of correctness determination detects the incorrect node for the first time.
Then, the naive algorithm backtracks from ${{A}}_n$, and continues traversing along [${{A}}_0,{{A}}_1,...,{{B}}_n$].
Assume that $A_0$ has two children.
According to our theorem, the probability of ${{A}}_1$ being a supporting lemma become less than $\frac{1}{2}$, \ie, $P(A_1)<\frac{1}{2}$, which was $\frac{1}{2}$ before the ${{A}}_n$ is identified as an incorrect node.
Similarly, the probability of all $A_1, A_2, \dots, A_{n-1}$ being supporting lemmata decreases exponentially.
Hence,  ${{B}}_n$ should not be the next traversed node due to the low probability.
In comparison, our DQN sets the rewards of nodes $A_1, A_2, \dots, A_{n}$ as negative values, such that the updated strategy in DQN tends to select $B_1$ instead of $A_1$.
The analysis is also validated in our experiments (See \Cref{subsec:mainevaluation}).

%% file: s-evaluation.tex
\section{Experiments \& Evaluation} 
\label{sec:evaluation}

We perform experiments on several security protocols.
The experiment results are described in \Cref{subsec:mainevaluation}.
In \Cref{subsec:casestudy}, we briefly overview Yubikey protocol to validate the efficiency of \ourtool.

\subsection{Main Experiments}
\label{subsec:mainevaluation}

\begin{table*}[t]
\caption{Experimental results on security protocols with unbounded sessions in verification tools.}
\center
\vspace{-0.1in}
\resizebox{1\textwidth}{!}{
\begin{tabular}{|c|c|c|c|c|c|c|c|c|c|c|c|c|c|c|c|c|c|c|c|c|}
\hline 
\multirow{4}{*}{Protocols} & \multirow{4}{*}{StatVerif} & 
\multirow{4}{*}{Set-$\pi$} & \multirow{4}{*}{GSVerif} & \multicolumn{10}{c|}{SAPIC/Tamarin Prover} &  
\multicolumn{7}{c|}{\ourtool} \\
\cline{5-21} &  & & & \multicolumn{2}{c|}{-`s'}  & \multicolumn{2}{c|}{-`c'} & \multicolumn{2}{c|}{-`p'} &  \multicolumn{2}{c|}{DFS}  & \multicolumn{2}{c|}{BFS}   & \multirow{3}{*}{\shortstack{Automated \\ Verification?}}  & \multirow{3}{*}{\shortstack{Training \\ Epochs}} & \multirow{3}{*}{Coverage} & \multicolumn{3}{c|}{Training Time} & \multirow{3}{*}{\shortstack{Verification \\ Time}} \\
\cline{5-14,18-20} &  & & &\multirow{2}{*}{Original} &\multirow{2}{*}{w/ DFS} & \multirow{2}{*}{Original} &\multirow{2}{*}{w/ DFS} & \multirow{2}{*}{Original} &\multirow{2}{*}{w/ DFS} & \multirow{2}{*}{Coverage} & Total  & \multirow{2}{*}{Coverage} & Total & &   &  & \multirow{2}{*}{Acquisition} & Network  & \multirow{2}{*}{Overall}  & \\
 & & & & & & & & & &  &  Time&  & Time &  &  &  &  & Training &  &  \\
\hline 
Yubikey  \cite{yubikeymanual} & $\times$ & $\times$ & $\times$ & $\times$  & $\times$  & $\times$ & $\times$ &  $\times$  & $\times$  &  N/A& $\times$ &  N/A & $\times$ & \checkmark  & 81 & 9253 & 41m & 39m  & 80m & 4m \\
\hline
M{\"{o}}dersheim's Example \cite{DBLP:conf/ccs/Modersheim10} & $\times$ & 7s & 6s & $\times$  & 2h & $\times$ &  4h & 6m  & 7m & 4424& 3h & 31037 & 7h & \checkmark  & 20  & 3064 & 10m & 8m   &  18m & 2m\\
\hline
Security API in PKCS\#11 \cite{PKCS11standard} & $\times$ & $\times$ & 15s & 25m  & 27m & 7h  & 9h & 18m &  20m& N/A& $\times$  & N/A& $\times$  & \checkmark & 164 & 11647 & 47m & 51m & 98m & 15m \\
\hline
GJM Contract-Signing \cite{DBLP:conf/crypto/GarayJM99} & 10s & $\times$ & 7s & 3m  & 4m & 1h & 1h & 2m & 2m & 6088 & 2h   & 33672& 5h & \checkmark & 14  & 2643& 8m & 9m & 17m & 2m  \\
\hline
one-dec \cite{GSVerif-CSF18} & $\times$ & $\times$ & 9s & $\times$  & 20h & $\times$  & 27h & $\times$  & 23h & 127982& 14h   &  N/A& $\times$  & \checkmark  & 30 & 3985& 12m & 15m  & 27m & 2m \\
\hline
one-dec,  table  variant \cite{GSVerif-CSF18} & $\times$ & $\times$ & 8s & 4m  & 6m & 1h  & 2h & 3m  & 5m & 55902& 8h   & 155378 & 14h  & \checkmark & 28  & 3324& 13m & 14m  & 27m & 2m \\
\hline
private-channel \cite{GSVerif-CSF18} & $\times$ & $\times$ & 7s & 3m  & 5m & 50m  & 1h & 3m  & 4m & 3224 & 3h  & 34683 & 5h  & \checkmark  & 24 & 2834 & 10m & 11m  & 21m & 2m \\
\hline
counter \cite{GSVerif-CSF18} & $\times$ & $\times$ & 9s & 3m  & 4m & 45m & 1h & 3m & 5m & 4416  & 3h  & 32537 & 6h & \checkmark  & 27 & 2967 & 11m & 11m  & 22m & 2m \\
\hline
voting \cite{GSVerif-CSF18} & $\times$ & $\times$ & 7s & 5m  & 6m & 41m  &  1h& 4m & 6m  & 156673   & 17h &  N/A& $\times$  & \checkmark  & 33 & 5464& 18m & 17m  & 35m & 3m \\
\hline
TPM-envelope \cite{DBLP:conf/csfw/DelauneKRS11} & $\times$ & $\times$ & $\square$  & $\times$  & 14h & $\times$ & 29h & $\times$  & 10h & 36499& 11h  & 206589 & 20h & \checkmark & 31  & 4067 & 14m & 15m  & 29m & 2m\\
\hline
TPM-bitlocker \cite{DBLP:conf/csfw/DelauneKRS11} & 5s & 6s &$\square$ & $\times$  & 1h & $\times$  & 5h & $\times$ & 2h & 4431 & 3h & 33248 & 6h  & \checkmark & 33  & 3075& 11m & 10m  & 21m & 2m\\
\hline
TPM-toy \cite{DBLP:conf/csfw/DelauneKRS11} & $\times$ & $\times$ &$\square$ & $\times$   & 2h & $\times$  & 6h & $\times$ & 4h & 4865  & 3h   & 23502& 6h & \checkmark  & 20 & 2521& 7m & 8m  & 15m & 2m  \\
\hline
Key registration \cite{DBLP:conf/csfw/BruniMNN15} & 4s & 6s & 7s & 2m  & 3m & 30m & 1h & 3m & 4m & 4048 & 2h   & 43322& 6h & \checkmark& 16  & 3087 & 7m & 7m  & 14m & 3m \\
\hline
Secure device \cite{DBLP:conf/csfw/ArapinisRR11} & $\times$ & 6s & 7s & $\times$ & $\times$ & $\times$ & $\times$ & 3m  & 4m & 86719& 13h  & N/A& $\times$ & \checkmark & 35  & 3991 & 10m & 12m & 22m & 1m\\
\hline
MaCAN \cite{DBLP:conf/ifm/BruniSNN14}  & $\times$  &  11s  &  14s&   2m & 2m & 27m  & 29m & 3m & 4m  & 4910 & 3h  &  37811  & 6h  & \checkmark& 34 & 3502  & 15m & 16m  & 31m & 2m   \\
\hline 
CANauth \cite{van2011canauth} & $\times$ & $\times$ & $\times$ & $\times$  & 6h & $\times$  & 13h &$\times$  & 10h & 6391  & 7h     & N/A& $\times$ & \checkmark& 31  & 3391  & 9m & 11m  & 20m & 2m \\
\hline
CANauth simplified \cite{van2011canauth} & $\times$ & $\times$ & $\times$ & $\times$  & 3h & $\times$  & 10h& $\times$ & 5h & 4344  & 4h &  N/A& $\times$  & \checkmark  & 39 & 3301 & 9m & 10m  & 19m & 2m \\
\hline
Mobile EMV \cite{DBLP:conf/eurosp/CortierFFGT17} & $\times$ & $\times$ & 6s & $\times$ & $\times$ & $\times$ & $\times$ & $\times$ & 25h & 239970 & 17h & N/A &$\times$  & \checkmark & 36  & 6013& 21m & 23m  & 44m & 1m \\
\hline
Scytl Voting System \cite{DBLP:conf/eurosp/CortierGT18} & $\times$ & $\times$ & 5s & 4m  & 5m & 47m   & 1h & 4m & 6m &  276673  & 21h & N/A & $\times$ & \checkmark  & 41 & 7027& 29m & 31m  & 60m & 2m \\
\hline
Chaum's Online e-Cash \cite{DBLP:conf/crypto/Chaum82} & $\times$ & $\times$ & $\times$ &  3m & 4m &  31m  &  1h& 4m& 6m & 3907  & 2h & 30231 & 6h & \checkmark  & 20 & 3122 & 13m & 15m  & 28m & 2m \\
\hline
FOO Voting  \cite{DBLP:conf/asiacrypt/FujiokaOO92} & $\times$ & $\times$ & $\times$ &  5m &  7m &  49m &1h & 5m& 6m &  4912 & 3h & 40215 & 6h & \checkmark  & 24 & 4245 & 15m & 16m  & 31m & 3m \\
\hline
Feldhofer's RFID \cite{DBLP:conf/ches/FeldhoferDW04} & $\times$ & $\times$ & $\times$ &  8m & 9m &  51m & 1h& 7m& 10m & 4893 & 3h & 39023 & 6h & \checkmark & 34 & 5753 & 18m & 21m  & 39m & 4m \\
\hline
Denning-Sacco  \cite{DBLP:journals/jcs/DelauneKR09} & $\times$ & $\times$ & $\times$ & 2m  & 3m &  29m &40m & 2m&4m &  3898 & 2h & 32459 & 5h & \checkmark & 26 & 5124 & 18m & 18m  & 36m & 3m \\
\hline
Okamoto \cite{DBLP:conf/ifm/RamsdellDGR14}  & $\times$ & $\times$ & $\times$ & 13m  & 15m &  2h & 3h& 14m& 20m & N/A& $\times$ &  N/A & $\times$ & \checkmark & 148 & 10931 & 41m & 47m  & 86m & 10m \\
\hline
\end{tabular}
	}
\begin{tabular}{c}
$\times$ no automatic verification (computation time $>$48h, memory used $>$128GB)\ \ \ \ \\
$\square$ requiring optimizing heuristics to achieve automatic verification  \ \ \ \  \checkmark automatic verification
\end{tabular}
\label{tab:benchmark}
\vspace{-0.1in}
\end{table*}



We compare \ourtool with other verification tools in verifying security protocols.
We evaluate the efficiency  of \ourtool.

\textbf{Experimental Setup:} Experiments are carried out on a server with Intel Broadwell E5-2660V4 2.0GHz CPU, 128GB memory and four GTX 1080 Ti graphic cards running Ubuntu 16.04 LTS.
We use and modify tamarin prover v1.4.0 in \ourtool.

We use the same network architecture, learning algorithm and parameter settings across all chosen protocols.
Since the security property varies greatly in protocols, we set all the negative rewards to -10 for generality.
In these experiments, we use the DQN with 0.01 learning rate and memory batch of size 7000.
Moreover, we execute eight threads of Acquisition module in parallel.
The behavior policy during training was $\epsilon$-greedy with $\epsilon$ annealed linearly from 0.99 to 0.1 over the first hundred epochs, and fixed at 0.1 thereafter.



\textbf{Chosen Tools:} For each protocol with unbounded sessions, we inspect whether it can be automatically verified by \ourtool and other verification tools.
These verification tools include StatVerif~\cite{DBLP:conf/csfw/ArapinisRR11},~Set-$\pi$ \cite{DBLP:conf/csfw/BruniMNN15}, tamarin prover~\cite{DBLP:conf/sp/KremerK14,DBLP:conf/cav/MeierSCB13} and GSVerif \cite{GSVerif-CSF18}.
The tools are typical verification tools which support verification of security protocols with global states.
Moreover, all these tools provide automated verification modes to verify security protocols.
Besides, we attempt to verify protocols in four altered versions of tamarin prover.
We first attempt to use the `c' heuristic of tamarin prover.
This heuristic adopts a simplest method to verify protocols: it solves goals in the order they occur in the constraint system.
Unlike other default static strategies, this method does not contain any human-designed heuristics or expertise.
We compare \ourtool with this mode of tamarin prover to demonstrate the generality of our framework.
Then, we try to verify protocols using the dedicated heuristic (\ie, the `p' heuristic), which is designed by the SAPIC authors \cite{DBLP:conf/sp/KremerK14} to efficiently solve SAPIC generated Tamarin models.
Since we use tamarin protocol models as well as SAPIC generated models in our experiment, we additionally compare \ourtool with this heuristic to validate the efficiency of our framework.
Moreover, we implement two naive algorithms (DFS and BFS) with our loop detection method and compare these two with our algorithm to show the efficiency of \ourtool.
Besides, we combine the DFS with the heuristics of tamarin prover to further validate the efficiency.
Note that we do not choose classical verification tools, \ie, ProVerif and AVISPA,
since their support for protocols with global states and unbounded sessions is limited.

\textbf{Chosen Protocols:} We carefully choose security protocols to be testified in our evaluation.

\textbf{1)} We choose all the protocols that have been evaluated in papers of the compared tools, \ie, StatVerif \cite{DBLP:conf/csfw/ArapinisRR11}, Set-$\pi$ \cite{DBLP:conf/csfw/BruniMNN15}, GSVerif \cite{GSVerif-CSF18}, tamarin prover \cite{DBLP:conf/cav/MeierSCB13}, SAPIC \cite{DBLP:conf/sp/KremerK14}.
The chosen protocols include a simple security API similar to PKCS\#11 \cite{PKCS11standard}, the Yubikey security token \cite{yubikeymanual}, the optimistic contract signing protocol by Garay, Jakobsson and MacKenzie (GJM) \cite{DBLP:conf/crypto/GarayJM99}, \etc
These protocols are typical protocols with global states, unbounded sessions.
Many research efforts
\cite{DBLP:conf/csfw/ArapinisRR11,DBLP:conf/csfw/BruniMNN15,DBLP:conf/sp/KremerK14,DBLP:conf/cav/MeierSCB13}
were spent on verification of these protocols.
Besides, GSVerif paper evaluated the performance of 18 protocols, which are all chosen in our paper.
In these protocols, Yubikey is the most important case for evaluation for it is most widely studied, but still have not been automatically verified, according to the current literature \cite{GSVerif-CSF18,DBLP:conf/csfw/BruniMNN15,DBLP:conf/stm/KunnemannS12}.

\textbf{2)} Since the security property of observation equivalence \cite{DBLP:conf/crypto/Chaum82,DBLP:conf/asiacrypt/FujiokaOO92,DBLP:journals/jcs/DelauneKR09,DBLP:conf/ifm/RamsdellDGR14,DBLP:conf/ches/FeldhoferDW04} cannot be verified by StatVerif, Set-$\pi$, or GSVerif while only tamarin provers supports verifying the property, we choose 5 protocols with the properties from the official repository \cite{tamarin-rep} of tamarin.
Specifically, these protocols include Chaum's Online e-Cash \cite{DBLP:conf/crypto/Chaum82}, FOO Voting \cite{DBLP:conf/asiacrypt/FujiokaOO92}, Denning-Sacco \cite{DBLP:journals/jcs/DelauneKR09}, Okamoto \cite{DBLP:conf/ifm/RamsdellDGR14}, and Feldhofer's RFID protocol \cite{DBLP:conf/ches/FeldhoferDW04}.

\textbf{3)} For fairness we do not choose other protocols. 
Practically there are many other protocols that cannot be automatically verified by state-of-the-art tools,  \eg, TLS \cite{DBLP:conf/sp/CremersHSM16}, and smart contract protocols \cite{notsosmart,erc20}.
Note that support lemmata have to be the manually specified to help prove TLS \cite{DBLP:conf/sp/CremersHSM16} in tamarin prover.
In comparison, \ourtool successfully verifies all the protocols automatically.
However, we do not compare \ourtool with existing tools by using the protocols, since it becomes questionable whether the protocols are cherry-picked and whether some of the protocols can be verified by customized heuristics, \eg, three protocols verified by optimized heuristics of GSVerif in Table \cite{GSVerif-CSF18}.
Instead, since the protocols evaluated in the papers \cite{DBLP:conf/csfw/ArapinisRR11,DBLP:conf/csfw/BruniMNN15,GSVerif-CSF18,DBLP:conf/cav/MeierSCB13,DBLP:conf/sp/KremerK14} are thoroughly studied, the experimental results on the protocols are more convincing.


\textbf{Comparative Results:} The experimental results are summarized in \Cref{tab:benchmark}.
Compared with these verification tools, \ourtool is sufficient for generality and automation capability: it is able to verify all the given protocols with unbounded sessions, without any human intervention.

1) \textit{Generality}.
\ourtool achieves a 100 percent success rate in verifying the studied protocols, which outperforms all the other verification tools.
For example, StatVerif does not terminate and the verification fails encountering complicated protocols like security API in PKCS \#11 and mobile EMV protocol \cite{DBLP:conf/eurosp/CortierFFGT17}.
Set-$\pi$ fails in verifying TPM-envelope protocol \cite{DBLP:conf/csfw/DelauneKRS11} and some others with unbounded sessions.
Tamarin prover is effective in automatically verifying simple protocols with unbounded sessions, \eg, GJM Contract-Signing protocol and Security API in PKCS\#11.
Nevertheless, it does not achieve automated verification of Yubikey protocol, TPM-bitlocker, \etc, in any of the studied versions.
GSVerif outperforms the previous tools in generality but it still can not automatically verify complicated protocols such as Yubikey protocol.
Note that its heuristics in 3 cases are rewritten to achieve automation \cite{GSVerif-CSF18}, \ie, 11 cases cannot be automatically verified without the optimization. In comparison, our heuristic is designed without any human intervention.
Moreover, among the tools, only tamarin prover and \ourtool can handle protocols with observational equivalence properties.
They achieve successful verification of the five protocol cases.

2) \textit{Automation capability}. 
Currently, only \ourtool can fully automatically verify Yubikey and CANauth protocol~\cite{van2011canauth}.
For protocols which existing tools cannot automatically verify, GSVerif and tamarin prover provide an interactive mode for users to manually guide the verification.
In our experiments, we find that Yubikey and CANauth protocol can be verified by manually designing proof formulas using these tools.
In contrast, \ourtool fully automatically verifies these protocols, without any human intervention.
In \Cref{subsec:casestudy}, we briefly overview Yubikey protocol to demonstrate  the sufficiency of \ourtool in automation capability. 
We further overview CANAuth protocol in \Cref{sec:casestudy}.

\textbf{Efficiency and Overhead:} To evaluate the efficiency and overhead of \ourtool,
we collect statistics of the running time and training epochs in verification.
The running time contains two parts: 1) \textit{Training time}: the time spent in information acquisition and network training; 2) \textit{Verification time}: the time spent in verification after network convergence.
As presented in this paper, \ourtool is a novel and general framework to verify protocols.
For each protocol to be verified, it takes time to acquire information and train the network.
Once the DQN is sufficiently optimized according to the current protocol model, it can be directly used to verify the corresponding protocol, like the static strategies in existing tools.
Hence, we use the verification time to demonstrate the performance of \ourtool.
Besides, since \ourtool uses the DQN to select proof paths,
the efficiency of the DQN directly affects the performance and overhead of \ourtool.
Recall that the number of epochs denotes the times that the DQN is optimized with a new reward, which is related to the number of generated incorrect paths, if the protocol has not been successfully verified.
Therefore, we use the quantities of training epochs and time of DQN to evaluate the efficiency of our framework.

As demonstrated in \Cref{tab:benchmark}, the experimental results show that \ourtool verifies the studied protocols in a very efficient way.
For most protocols, it succeeds in verification only after about 25 times of one-way forward traversing (\ie, 25 epochs).
As the challenge is time explosion when traversing infinite state spaces, our dynamic strategy solves the problem in a general and adequate way.
For instance, existing verification tools can not automatically verify Yubikey protocol with unbounded sessions due to memory explosion or infinite verification loops.
In contrast, \ourtool only takes 81 epochs to find the correct proof path using the dynamic strategy.

Moreover, the statistics of the running time also validate the efficiency of \ourtool.
Comparing to existing verification tools, \ourtool does not require any extra time and effort in training human for interactive proving or designing heuristics.
Instead, it spends the training time on optimizing the dynamic strategy for protocols, which is sufficient.
Specifically, if a protocol's model is complicated, \ie, the searching space is large, the running time increases.
The space's size depends on whether the model covers global states or unbounded sessions \cite{yubikeymanual,van2011canauth,PKCS11standard}, or whether the model is simplified \cite{DBLP:conf/ccs/CremersHHSM17,DBLP:conf/sp/KremerK14}.
For most protocols, it only takes less than half an hour to find the correct proof path.
In the worst case, it costs 98 minutes to verify the security API in PKCS~\#11.
After the DQN is sufficiently optimized, the verification time is only 13 minutes.

\textbf{Performance Analysis:} Besides proving our insight theoretically, we also perform empirical analysis by comparing \ourtool with two naive algorithms: 
1) \textbf{DFS}. The algorithm searches along a path as long as possible before backtracking. 
The backtracking occurs only when a loop is detected.  
2) \textbf{BFS}. The algorithm searches all the paths at the present depth prior to searching at the next depth level.
It also uses the loop detection algorithm to shrink the size of searching space.
The BFS is optimized by multi-threading that each threads searches in parallel.
Note that DFS has to run in a single thread since the ordering and parallel tends to conflict in searching.
Both DFS and BFS are implemented based on tamarin prover.


\begin{figure}[tbp]
    \begin{minipage}[t]{0.45\linewidth}
        \centering
        \includegraphics[width=4.2cm,height=3.5cm]{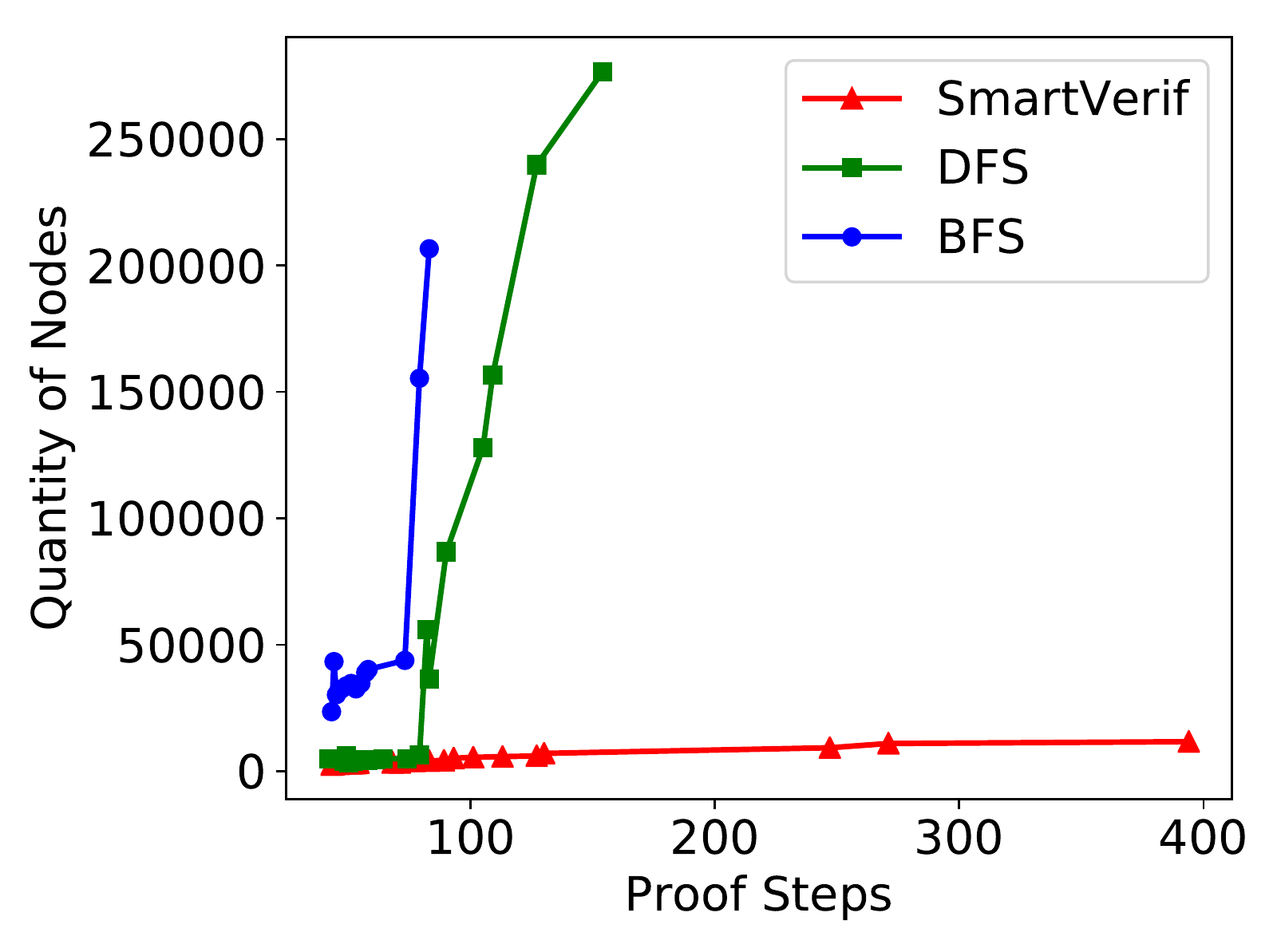}
        \parbox{4.5cm}{\scriptsize \hspace{1.2cm} (a) \# of covered nodes}
    \end{minipage}
    \hspace{1ex}   
    \begin{minipage}[t]{0.45\linewidth}
        \centering
        \includegraphics[width=4.3cm,height=3.5cm]{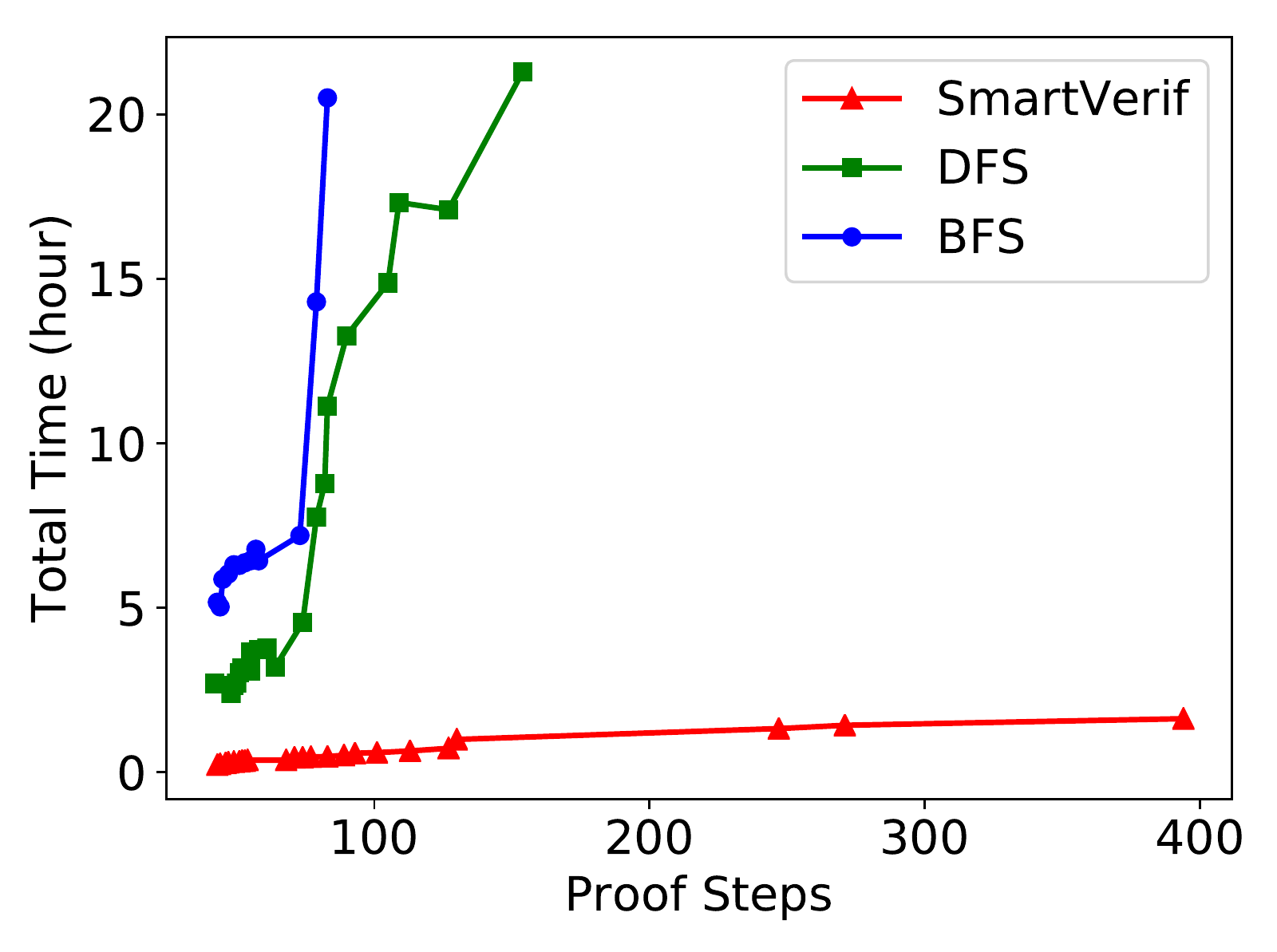}
        \parbox{4.5cm}{\scriptsize \hspace{1.5cm} (b) Total time}
    \end{minipage}
    \begin{center}
    \vspace{-0.2in}
    \caption{Experimental results. (a): Quantity of covered nodes compared with naive algorithms. (b): Total time compared with naive algorithms.}
    \label{fig:se}
    \end{center}
    \vspace{-0.4in}
\end{figure}

The experimental results are shown in \Cref{tab:benchmark} and \Cref{fig:se}.
We use two metrics: \textbf{1)} the total time in searching; \textbf{2)} coverage, \ie, the quantity of nodes that have been traversed when the searching succeeds, given the proof steps of verifying a security protocol.
Here, for \ourtool, the coverage metric includes the nodes traversed during the Acquisition and Verification phases.
Before comparison, an important observation is that it takes several seconds for a single step of new node traversing by using tamarin prover. 
It may take less time if using other tools, \eg, the ProVerif-based tools.
We also find that it takes more time when \textbf{a)} traversing a new node at the deeper level of tree, and \textbf{b)} initializing or reconfiguring the searching environment. 
For example, on verifying YubiKey protocol, the averaging time on traversing a node at level 10 and 100 is 1s and 2s, respectively, and the time on initialization is 4s. 
Therefore, the verification time of DFS and BFS tends to be affected by reason a) and b), respectively.
Since the verification time may be affected by multiple factors, we also use the number of traversed nodes as a complement metric in comparison.

A significant result is that DFS's coverage grows much faster than SmartVerif's coverage, when the proof steps increase starting from 65. 
Afterwards, the verification time of DFS reaches 48-hour limit when the proof steps are around 200.
We further find that for most protocols with proof steps less than 60, DFS only needs to backtrack for less than 10 steps.
For instance, for the TPM-toy protocol, DFS begins backtracking when it reaches the node at the depth 57, for the corresponding path is estimated incorrect. 
When succeeding in searching, the top 49 nodes in the incorrect path are the correct nodes representing supporting lemmata. 
Hence, when the depth for which DFS has to backtrack merely grows to more than 10, the performance of DFS starts to decrease drastically.

Therefore, SmartVerif greatly outperforms DFS when verifying complicated protocols. 
SmartVerif's coverage grows much slowly when the proof steps increase.
The phenomenon can be explained by our insight as illustrated in \Cref{subsec:strategy}.
Observe that the performance of BFS is even worse than the performance of DFS, though BFS runs in parallel. 
We omit the explanation due limitation of paper size.

Moreover, we implemented three naive algorithms as illustrated in \Cref{subsec:strategy}, which use the built-in heuristics (`s', `c' and `p') of tamarin prover as the static strategy of selecting nodes respectively, DFS for tree traversing, and our module of correctness determination for back-traversing.
As shown in \Cref{tab:benchmark}, the comparative results are summarized as follows.
\textbf{1)} For protocols like Yubikey, the naive algorithms still cannot succeed in automated verification. 
\textbf{2)} For protocols that cannot be verified by the original tamarin prover, \ourtool achieves much better efficiency compared with the naive algorithms.
\textbf{3)} For protocols that can be verified by the original tamarin prover, the naive algorithm only achieves similar performance with the original tamarin prover with the corresponding heuristics. 
\textbf{Discussion:} The results validates our analysis in \Cref{subsec:strategy}.
Here, an important observation is that for protocols of results 2), it is uncertain whether the naive algorithms with the built-in heuristics outperform the DFS without heuristics. 
An example is that Mobile-EMV protocol must be verified for at least 25 hours with the former algorithms, but it requires 17 hours for the latter algorithm.
It can be inferred that the design of static strategies is non-trivial:
an algorithm with a static strategy cannot be easily improved by leveraging other naive approaches, \eg, DFS.

\begin{figure}[tbp]
        \begin{minipage}[t]{0.45\linewidth}
            \centering
            \includegraphics[width=4.1cm,height=3.5cm]{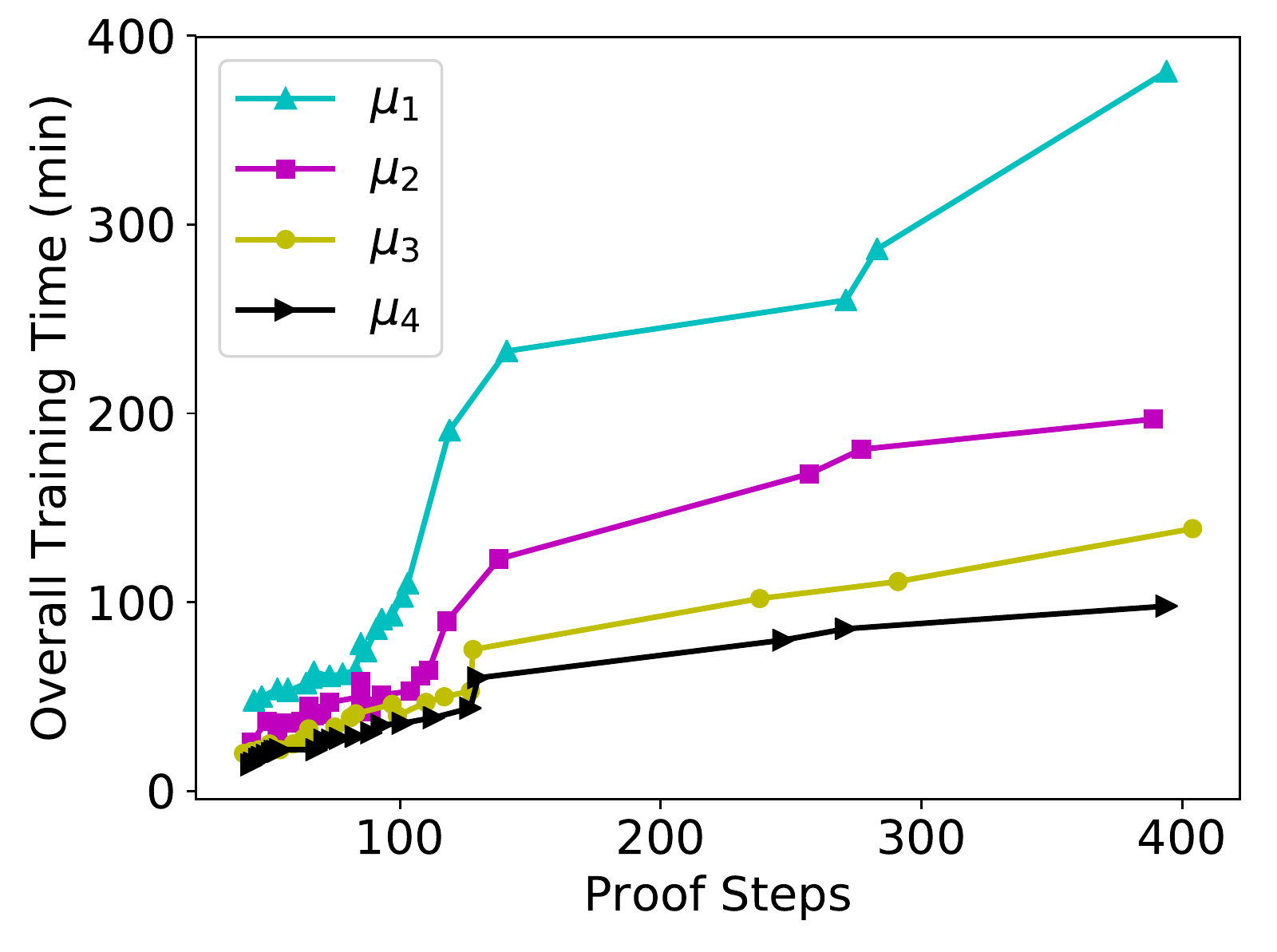}
            \parbox{4.5cm}{\scriptsize \hspace{1cm} (a) Overall training time}
        \end{minipage}
        \hspace{1ex}   
        \begin{minipage}[t]{0.45\linewidth}
            \centering
            \includegraphics[width=4.3cm,height=3.5cm]{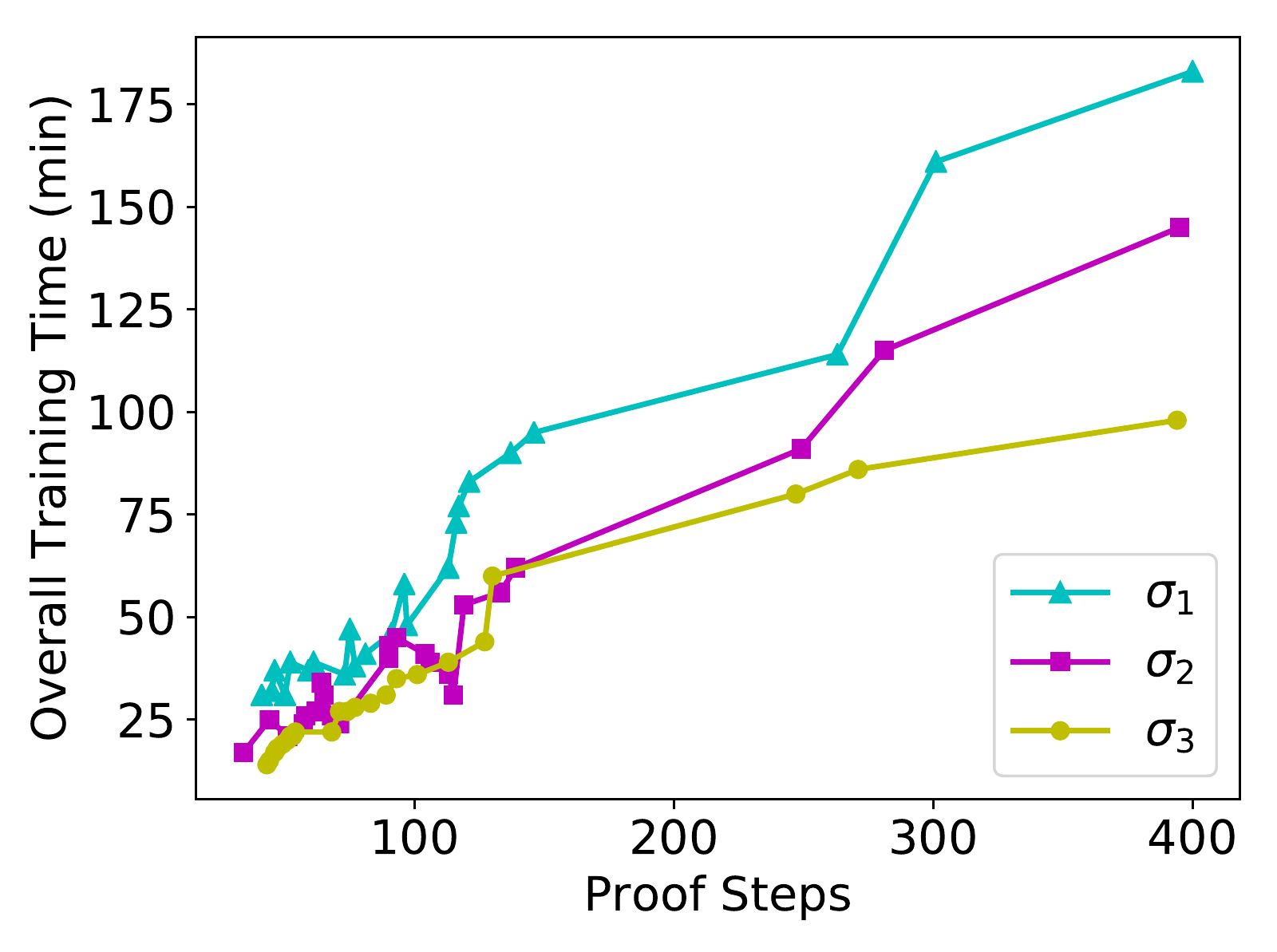}
            \parbox{4.5cm}{\scriptsize \hspace{1cm} (b) Overall training time}
        \end{minipage}
        \begin{center}
        \vspace{-0.1in}
        \caption{Experimental results. (a): Overall training time using different number of GPUs. (b): Overall training time using different multithreading parameters.}
        \label{fig:se2}
        \end{center}
        \vspace{-0.4in}
\end{figure}

In addition, we study the performance of training when using different number of GPUs.
We try four sets of parameters.
Here, $\mu_1$ represents running \ourtool without GPUs.
$\mu_2$ represents 0 graphic cards, which means that \ourtool only use the integrated GPU in the CPU to compute in the training process.
$\mu_3, \mu_4$ represents 2, 4 GTX 1080 Ti graphic cards respectively.
As shown in \Cref{fig:se2}~(a), we can see an improvement to the overall training times when using more graphic cards in our experiment.
For example, it only takes 80 minutes to verify the Yubikey protocol using four graphic cards.
Using no GPUs, it takes 260 minutes to achieve a successful verification.

Furthermore, we evaluate the overall training times in verifying four protocols with different multithreading parameters.
We try three sets of parameters.
$\sigma_1, \sigma_2, \sigma_3$ represents 2, 4, 8 threads of Acquisition module executed in parallel respectively.
As shown in \Cref{fig:se2} (b), the running time is decreasing with the increasing quantity of threads executed in parallel for the parameters sets $\sigma_1$, $\sigma_2$ and $\sigma_3$.
As shown in the above experimental results, \ourtool achieve a solid performance on a high-performance server as well as a modest machine with less graphic cards.

Note that we currently train a standalone DQN for each studied protocol to keep a high level of generality. 
Another possible approach is to use pre-trained and optimized networks to verify protocols. 
However, it brings several challenges.
Firstly, it is challenging to achieve a high level of accuracy on node selection in generating pre-trained network.
Existing works generating pre-trained networks \cite{Piotrowski_2018,Sekiyama2018AutomatedPS,alemi2016deepmath,kaliszyk2017holstep} in a similar research field, \ie, theorem proving, do not achieve a high level of accuracy on node selection.
Compared with theorem proving, it is much challenging to generate pre-trained network with much higher accuracy, given much less samples of models of security protocols.
Secondly, it is challenging to achieve high efficiency if using a generated pre-trained network.
If using a pre-trained network, 
 the verification time for some protocols may increase.
For example, one could take the standard heuristic of tamarin prover as the basic strategy in our DQN to verify security protocols.
However, in this case, the DQN does not optimize itself in an efficient way when verifying complicated protocols like Yubikey protocol.
Therefore, we train a standalone DQN for each studied protocol.
Similarly we currently retrain the DQN when verifying a new security property of a protocol.
We will try to optimize the network design and use other learning techniques in future work.

\subsection{Case Study}
\label{subsec:casestudy}

In the following, we briefly overview the Yubikey protocol \ourtool verified.
We provide some details in key steps of the verification.
For the limitation of paper size, we do not detail all the formal models of the protocols and properties that we studied.




Kremer \etal \cite{DBLP:conf/sp/KremerK14} modeled and verified Yubikey protocol with unbounded sessions in tamarin prover.
Specifically they define three security properties. 
All properties follow more or less directly from a stronger invariant.
By default, tamarin prover cannot automatically prove this invariant, which is caused by a non-termination problem.
To successfully verify the protocol, tamarin needs additional human guidance, which is provided by experts in the interactive mode.

In the following, we analyze the choice made by tamarin prover, experts and \ourtool.
Specifically, in proof step~\#8, tamarin prover needs to select one rule, \ie, lemma, from the rules as follows:

\resizebox{.778\linewidth}{!}{
\begin{minipage}{\linewidth}
\begin{align}
& A: (\#vr.13 < \#t2.1) \parallel  (\#vr.13 = \#t2.1) \parallel        (\#vr.6 <  \#vr.13) ) \nonumber \\
& B: State\_011111111111( lock11.1, n, n.1, nonce.1, npr.1, otc.1, \nonumber \\
& \ \ \ \ \ \ \ \ \ \  \ \  \ \  secretid, tc2, tuple  ) \triangleright_\textup{o} \#t2 ) \nonumber \\
& C: Insert( <'Server', n>, <n.2, n.1, otc>) @ \#t2.1 ) \nonumber \\
& D: !KU( n ) @ \#vk.2 ) \nonumber \\
& E: !KU( senc(<n.2, (otc+z), npr>, n.1) ) @ \#vk.5 )  \nonumber 
\end{align}
\end{minipage}
}

\vspace{0.1in}

Here, rule~$A$ is a restriction rule to the timepoints $\#vr.13, \#t2.1$ and $\#vr.6$.
Rule~$B$ states an action $State\_011111111111$ must have been in the protocol execution in timepoint $\#t2$.
Rule~$C$ states an action $Insert( <'Server', n>, <n.2, n.1, otc>)$ must have been in the protocol execution at timepoint $\#t2.1$.
Rule~$D$ states the adversary has known the nonce $n$ at timepoint $\#vk.2$.
Rule~$E$ states the adversary has known the encrypted message $senc(<n.2, (otc+z), npr>, n.1)$ at timepoint $\#vk.5$.

\begin{figure}[tbp]
\centering
\includegraphics[width=4.5cm,height=2.5cm]{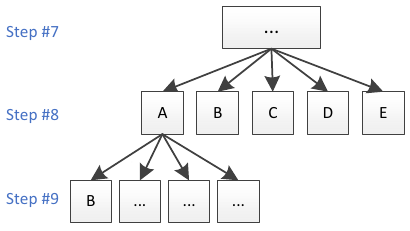}
\begin{center}
\vspace{-0.2in}
\caption{Part of the verification tree in Yubikey protocol}
\label{fig:performance4}
\end{center}
\vspace{-0.4in}
\end{figure}

Tamarin prover considers that rule $A$ is a timepoint constraint rule, which is more likely to achieve a successful verification.
It chooses the rule in the automated mode.
Then, there are four rules to be chosen.
By default, tamarin chooses the rule $B$.
However, the rule leads to a loop in verification as follows:
\resizebox{1\linewidth}{!}{
\begin{minipage}{\linewidth}
\begin{align*}
& State\_011111111111( lock11.1, n, n.1, nonce.1, npr.1,  otc.1, secretid, tc2, tuple  ) \triangleright_\textup{o} \#t2  \\
& ...\\
& Insert( <'Server', n>, <n.2, n.1, otc>) @ \#t2.1 )  \\
& State\_0111111111111( lock11.2, n, n.1, nonce.2, npr.2,otc.2, ~n.2, otc, tuple   ) \triangleright_\textup{o} \#t2.1  \\
& ...\\
& Insert( <'Server', n>,<n.2, n.1, otc.1>) @ \#t2.2 ) \\
& State\_0111111111111( lock11.3, n, n.1, nonce.3, npr.3, otc.2, ~n.2, otc.1, tuple ) \triangleright_\textup{o} \#t2.2  \\
\end{align*}
\end{minipage}
}

In this loop, tamarin prover keeps solving $Insert(<'Server', n>, <n.2, n.1, otc>) @ \#t2.1 )$ and $State\_0111111111111( lock11.2, n, n.1,nonce.2, npr.2,\\
 otc.2,n.2, otc, tuple   ) \triangleright_\textup{o} \#t2.1$ rules alternately.
It leads to non-termination in verification.

In interactive mode, experts make 23 manual rule selections to verify the protocol,
and 11 of them are different from the one made by tamarin prover.
Specifically, experts choose rule $B$ as the supporting lemma at proof step \#8,
which leads to a successful verification.

In \ourtool, we achieve a fully automated verification of Yubikey protocol without any user interaction.
\Cref{fig:performance4} shows the corresponding part of the verification tree.
The Q value of each rule in proof step \#8 is shown in \Cref{tab:yubikey}.
In the initial epoch, the Q value of each rule is the same.
In epoch 20, the network learns from its experience that candidate rules~$A,C,D,E$ may lead to non-termination cases with higher probability.
Hence, the Q values of these rules have a slighter difference compared with Q value of rule~$B$.
Then, the difference between Q value of rule $B$ and the Q value of other rules is getting larger in further epochs, which also validate our insight and the effectiveness of our designed strategy.
In epoch 81, \ourtool finds a correct proof path when choosing rule~$B$.
In further epochs, the difference among Q value of each rule is getting larger.
Based on the Q values, \ourtool finds the supporting lemma $B$ automatically, such that the
protocol can be verified without any user interaction.

\begin{table}[tb]
        \caption{Q value of each rule in proof step \#8.}
        \label{tab:yubikey}
        \centering
\vspace{0.1in}

\resizebox{0.4\textwidth}{!}{
\begin{tabular}{|c|c|c|c|c|c|}
\hline 
 & rule A &  rule B  &   rule C  &  rule D  &  rule E \\
\hline 
 initial epoch & 0 & 0 & 0 & 0 & 0  \\
\hline 
 epoch 20 & 0.3 & 0.4 & 0.2 & 0.2 & 0.3 \\
\hline 
 epoch 40 & 0.4 & 1.0 & 0.5 & 0.5 & 0.5 \\
\hline 
 epoch 81 & 1.2 & 1.6 & 1.2 & 1.3 & 1.0 \\
\hline
 epoch 120 & 1.2 & 1.9 & 1.3 & 1.3 & 1.1 \\
\hline
\end{tabular}
}
\vspace{-0.1in}
\end{table}

%% file: s-futurework.tex
\section{Future Work} 
\label{sec:futurework}

Our work opens several directions for future work.
\textbf{1)}~\textit{Hybrid strategy}. 
Since the initial strategy in \ourtool is purely random, the strategy may be optimized with less epochs if it is implemented with some static strategy.
However, the problem is still challenging that there is a potential risk that the epochs may become larger for some special protocols that the static strategy does not support.
\textbf{2)}~\textit{Scalability}. It is possible that our dynamic strategy can be used to cope with more complicated problems, such as automated formal verification of software or systems~\cite{DBLP:conf/sp/MurrayMBGBSLGK13,DBLP:conf/ccs/CockGMH14} that are based on first-order logics \cite{oueslati2017distributed,DBLP:journals/sigsoft/Romanovsky12} or higher-order logics \cite{DBLP:conf/csfw/BartheBCL12,DBLP:conf/types/BartheBCCL13}.
They are quite similar that they can be translated into a path searching problem.
We will also explore and verify more complicated security protocols using \ourtool.
\textbf{3)}~\textit{Efficiency}. Currently, we train a standalone DQN for each studied protocol to keep a high level of generality.
Designing a universal network which can verify all the protocols may increase the efficiency and improve the performance of \ourtool.
Therefore, we will try to optimize the network design and use other AI techniques in future work.


%% file: s-conclusion.tex
\section{Conclusion} 
\label{sec:conclusion}

In this paper we have studied automated verification of security protocols.
We propose a general and dynamic strategy to verify protocols.
Moreover, we implement our strategy in \ourtool, by introducing a reinforcement learning algorithm.
As demonstrated through experiment results, \ourtool automatically verifies security protocols that is beyond the limit of existing approaches.
The case study also validates the efficiency of our dynamic strategy.

%% file: s-a-proof.tex

\subsection{Proof of Our Insight} 
\label{sec:proof_of_our_insight}
We prove our insight of the paper that the node representing a supporting lemma is on the incorrect path with lower probability, when a random strategy is given.
To illustrate our insight more comprehensively, we translate the complicated verification process into a path searching problem.
Here, the verification can be simply regarded as the process of path searching in a tree:  each node represents a proof state which includes a lemma as a candidate used to prove the lemma in its father.
The supporting lemma is a special lemma necessarily used for proving the specified security property.


Formally, suppose there are $R$ correct and complete proof paths in a given tree, denoted as $[n_{t_1}, n_{t_2}, ..., n_{t_{k_t}}]$, where $n_{t_i}$ is the $i$th proof state at the $t$th path.
Therefore, the lemmata in $\{n_{t_i}\}$ are the candidate lemmata. 
The random strategy here means that whenever choosing a child for searching, the probability of choosing is uniform.
In other words, the probability of choosing the child $n_{t_i}$ is $\frac{1}{x_{t_i}}$.
If $n_{t_i}$ has $x$ children, and $y$ of them represent supporting lemma, and the random strategy is applied in choosing child, then the probability of choosing a nodes representing supporting lemma is $\frac{y}{x}$.
Suppose there are at least one child of $n_{t_i}$ that does not represent supporting lemma, \ie, $x>y$.

\newtheorem{mydef}{Theorem}
\begin{mydef}
Given the above assumptions, after $n_{t_i}$ has been chosen, 
the node representing a supporting lemma, who is the child of $n_{t_i}$, is on an incorrect path with the probability less than $\frac{y}{x}$.
\end{mydef}
\begin{proof}
For the $r$th path, define $\alpha_r$ as follows:
$$
\alpha_r=
\begin{cases}
\prod_{j=i+1}^{k_r}\frac{1}{x_{r_j}} & \textrm{if}\  \forall  j \in [1, i]. n_{r_j}=n_{t_j} \\
0 & \textrm{otherwise}
\end{cases}
$$

Denote $p_1$ as the probability that a selected path is incorrect. 
$$p_1=1-\sum_{r=1}^{R} \alpha_r$$

Denote $m_1, m_2, \dots, m_y$ as the nodes representing supporting lemmata among the children of $n_{t_i}$. 
For the $r$th path, define $\beta_{r,s}$ as follows:
$$
\beta_{r,s}=
\begin{cases}
\prod_{j=i+2}^{k_r}\frac{1}{x_{r_j}} & \textrm{if}\  n_{r_{i+1}}=m_s \land \forall  j \in [1, i]. n_{r_j}=n_{t_j} \\
0 & \textrm{otherwise}
\end{cases}
$$

It can be inferred that
$$
\alpha_r=\sum_{j=1}^y \frac{1}{x}\beta_{r,j}
$$

Denote  $p_2$ as the probability that a selected path is incorrect and the child representing supporting lemma is on the path. 
$$p_2=\sum_{j=1}^y \frac{1}{x} (1-  \sum_{r=1}^R\beta_{r,j})$$
Therefore, denote $p$ as the probability that a child representing supporting lemma is on an incorrect path.

$$
p=\frac{p_2}{p_1}= \frac{\sum_{j=1}^y \frac{1}{x} (1-  \sum_{r=1}^R\beta_{r,j})}{1-\sum_{r=1}^{R}\sum_{j=1}^y \frac{1}{x}\beta_{r,j}}=\frac{y-\sum_{r=1}^{R}\sum_{j=1}^y \beta_{r,j}}{x-\sum_{r=1}^{R}\sum_{j=1}^y \beta_{r,j}}<\frac{y}{x}
$$



\end{proof}





As a result, given a random strategy, the probability that $n_i$ is on an incorrect path is less than the probability that $n_i$ is on a given path.
In other words, if an incorrect path is found, the probability that $n_i$ is on the path, which equals $\prod_{j=1}^i\frac{1}{x_j}$ on a given path, decreases. 
On the other hand, the DQN requires a reward for guiding the optimization, where a reward corresponds to a determined occurrence of an event, \eg, a dead-or-alive signal upon an action in a game \cite{DBLP:journals/corr/MnihKSGAWR13}. 
However, there is no such determined event in verifications.
Instead, in \ourtool, we leverage probability of occurrence that the node representing a supporting lemma is on incorrect paths for constructing the reward according to Theorem 1. 
This insight enables us to leverage the detected incorrect paths to guide the path selection, which is implemented by using the DQN.


%% file: s-a-technique.tex

\subsection{Technical Details - Deep Q Network} 
\label{sec:technique}

\Cref{alg:dql} demonstrates the technical details of our implementation of DQN.
The DQN runs iteratively with multiple epochs.
In each epoch, recalling that we adopt a multi-threading approach for increasing the efficiency, the DQN launches $\sigma$ threads in which the paths are selected according to the policy (line~5).
If a path is estimated correct and complete, SmartVerif terminates with the proof path (line~6).
If all the selected paths are estimated incorrect, the policy is optimized (line~7). 

In path selection, we use two strategies in the policy~(line~12):
1) an exploration strategy to choose random actions, which is to explore the values of unchosen actions;
2) a greedy strategy to choose $a$ which may have the largest~$Q$ value currently.
Here, $Q(s_t, a; \theta_e)$ is a pre-defined function~\cite{DBLP:journals/corr/MnihKSGAWR13} that outputs comparable value, given the node $s_t$ and its $a$th child.
The Q function also takes $\theta_e$ as input, where~$\theta_e$ is the set of the DQN's parameters at epoch $e$, and~$\theta_e$ is updated into $\theta_{e+1}$ in policy optimization.
Combining the two strategies, we use a $\epsilon$-greedy strategy to select actions. Here, $\epsilon$ is a probability value for selecting random actions.
We change the value of $\epsilon$ to get different exploration ratios.
Note that we choose random actions in the exploration strategy.
Another possible approach is to take standard heuristic of tamarin prover as the basic strategy.
However, for example, when verifying Yubikey protocol, the standard heuristic does not rank the supporting lemma at the first place in several proof steps.
In this case, the DQN does not optimize itself in an efficient way and the efficiency is worse than \ourtool.


\begin{algorithm}[ht]
\caption{Implementation of DQN}  
\label{alg:dql}  
\algrenewcommand{\ALG@beginalgorithmic}{\footnotesize}
\begin{algorithmic}[1] 
\State Initialize a replay memory $D$ to capacity $N$ 
\State Initialize an action-value function $Q$ 
\State $success = 0$ 
\For{$e=1$ to \textit{EPOCH}}
	\State{Call $\sigma$ threads that execute path\_selection }
	\State \textbf{if} $success = 1$ \textbf{then} Program ends
	\State{Execute policy\_optimization}
\EndFor

\State{}
\State{\textbf{function} path\_selection:}
\State{Initialize a proof state $s_1$}
\For{$t=1$ to \textit{ROUND}}
\State With probability $\epsilon$ select a random action $a_t$ 
\Statex{ \ \ \ \ \ \  \  otherwise select $a_t = max_a Q(s_t, a; \theta_e)$}
\State{Generate next state $s_{t+1}$ according to $a_t$}
\State{Store a transition $(s_t, a_t, \omega, s_{t+1})$ in $D$}
\State \textbf{if} the path is estimated incorrect \textbf{then break}
\If{the path is estimated correct and complete}
\State{$success = 1$}
\State{\textbf{return}}
\EndIf
\EndFor
\State{}
\State{\textbf{function} policy\_optimization:}
\State{Sample $n$ random transitions $(s_j, a_j, r_j, s_{j+1})$ from $D$}
\State{Set $y_j = r_j + \gamma max_{{a}'}Q(s_{j+1},{a}';\theta_e)$}
\State{Perform a gradient descent step on $(y_j - Q(s_j, a_j; \theta_{e+1}))^{2}$}
\end{algorithmic}  
\end{algorithm} 

To apply our insight, we set the reward to the same negative number for all the edges on each estimated incorrect proof path. 
Specifically, in line 14, a transition, \ie, tuple~$(s_t,a_t, \omega, s_{t+1})$, is generated and added to $D$, where $w$ is the negative reward for the action $a_t$ at the state $s_t$.
$D$ is a replay memory  \cite{lin1993reinforcement} with capacity $N$, \ie, 
in practice, our network only stores the last $N$ tuples in the replay memory.

In policy optimization, $\theta_e$ in $Q$ function is updated as mentioned (line 20). 
Here, $n$ tuples are randomly selected from~$D$.
For each selected tuple $(s_j,a_j,r_j,s_{j+1})$, we compute $y_i$ according to $\theta_e$.
Then $\theta_{e+1}$ is estimated by using the loss function $(y_i- Q(s,a,\theta_{e+1}))^{2}$.

%% file: s-a-casestudy.tex

\subsection{Case Study - CANAuth protocol} 
\label{sec:casestudy}

We also investigate the case study presented by CANAuth protocol.
Cheval \etal \cite{GSVerif-CSF18} encoded a model for the protocol.

In the following, we analyze the choice made by tamarin prover, human experts and \ourtool.
In proof step \#10, tamarin prover needs to select one rule from the following rules:

\resizebox{.95\linewidth}{!}{
\begin{minipage}{\linewidth}
\begin{align*}
& A: solve( (\#vr.29 < \#t2.1) | (\#vr.29 = \#t2.1) )   \\
& B: solve( Insert( n.5, i ) @ \#t2.1 ) \\
\end{align*}
\end{minipage}
}

Rule $A$ states that timepoint $\#vr.29$ is earlier than or equals to $\#t2.1$.
Rule $B$ states action $Insert( n.5, i )$ is executed at timepoint $\#t2.1$.

Since the strategy of tamarin prover decides that the second rule is unlikely to result in a contradiction,
it chooses rule $A$ in the automated mode.
However, the rule leads tamarin prover to a loop as follows:
\resizebox{.8\linewidth}{!}{
\begin{minipage}{\linewidth}
\begin{align*}
& solve( State\_0111111111211111( lock7, n.5, cellB, i, msg.1, sk ) \triangleright_\textup{o} \#t2.1 )  \\
& solve( State\_0111111111211111( lock8, n.6, cellB, i, msg.2, sk ) \triangleright_\textup{o} \#t2.2 )  \\
& solve( State\_0111111111211111( lock9, n.7, cellB, i, msg.3, sk ) \triangleright_\textup{o} \#t2.3 ) \\
& ...\\
\end{align*}
\end{minipage}
}

In interactive mode, experts make 4 manual rule selections to verify the protocol,
and one of them is different from the selection made by tamarin prover.
Specifically, experts choose rule $B$ in proof step \#10, which leads to success of the verification.

In \ourtool, the result is similar to the previous case.
\Cref{fig:performance5} shows the corresponding part of the verification tree.
The Q value of each rule at proof step \#10 as shown in \Cref{tab:moderheim}.
In the initial state, the Q value of each rule is the same.
In epoch 10, the DQN discovers that candidate rule~$A$ may lead to incorrect paths.
Hence, the Q value of rule~$A$ has a slighter difference compared with rule~$B$.
Then, in epoch 23, \ourtool finds a correct proof path when choosing rule~$B$.
In epoch 100, the difference continues increasing.


\begin{figure}[htbp]
\centering
\includegraphics[width=4cm,height=3cm]{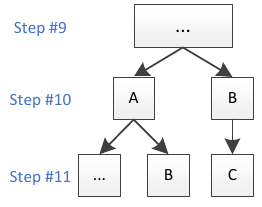}
\begin{center}
\caption{Part of the verification tree in CANAuth protocol.}
\vspace{-0.2in}
\label{fig:performance5}
\end{center}
\end{figure}

\begin{table}[h]
        \caption{Q value of each rule in proof step \#10.}
        \label{tab:moderheim}
        \centering
\resizebox{0.20\textwidth}{!}{
\begin{tabular}{|c|c|c|}
\hline 
 & rule A &  rule B   \\
\hline 
 initial epoch & 0 & 0   \\
\hline 
 epoch 10 & 0.3 & 0.5  \\
\hline 
 epoch 23 & 1.2 & 1.5  \\
\hline 
 epoch 100 & 2.7 & 5.1  \\
\hline
\end{tabular}
}
\end{table}